\documentclass[sigconf]{acmart}
\usepackage[utf8]{inputenc} % allow utf-8 input
\usepackage[T1]{fontenc}    % use 8-bit T1 fonts
\usepackage{booktabs}       % professional-quality tables
\usepackage{nicefrac}       % compact symbols for 1/2, etc.
\usepackage{microtype}      % microtypography
\usepackage[dvipsnames]{xcolor}         % colors
\usepackage{amsthm}
\usepackage{amsmath}

\usepackage{wrapfig}

\usepackage{amsthm}
\usepackage{hyperref}

\usepackage{glossaries}
\usepackage{url}
\usepackage{siunitx}
\usepackage{graphicx}
\usepackage{tcolorbox} 

\usepackage{algorithm}
\usepackage{algorithmic}
\usepackage{multicol}
\usepackage{multirow}
\usepackage{pifont}

\usepackage{caption}
\usepackage{subcaption}

\usepackage{enumitem}

\DeclareMathOperator*{\argmin}{arg\,min}

\newcommand{\pw}{\mathbf{w}}
\newcommand{\px}{\mathbf{x}}
\newcommand{\pu}{\mathbf{u}}

\newcommand{\reals}{\mathbb{R}}
\newcommand{\uspace}{\mathbb{U}}
\newcommand{\conts}[1]{\mathcal{C}^\omega(#1)}

\newcommand{\divergence}{\nabla \cdot}

\newcommand{\vol}[1]{\mathrm{Vol}\big(#1\big)}
\newcommand{\toi}{\tau}
\newcommand{\cmass}{\Omega_\toi}
\newcommand{\volf}{V_\Omega}
\newcommand{\volfapprox}{\tilde{V}_\Omega}
\newcommand{\volfbound}{\bar{V}_\Omega}
\theoremstyle{definition}
\newtheorem{definition}{Definition}
\newtheorem{problem}{Problem}

\theoremstyle{plain}
\newtheorem{theorem}{Theorem}
\newtheorem{corollary}[theorem]{Corollary}

\newtheorem{lemma}[theorem]{Lemma}

\theoremstyle{plain}
\newtheorem{remark}{Remark}

\usepackage{eso-pic}

\begin{document}

%\AddToShipoutPictureBG*{%
%  \AtPageUpperLeft{%
%    \hspace{16.5cm}%
%    \raisebox{-1.5cm}{%
%      \makebox[0pt][r]{To appear in ACM Int'l Conf. on Hybrid Systems: Computation and Control (HSCC), and ACM/IEEE In'l Conference on Cyber-Physical Systems (ICCPS) 2026}}}}

%\AddToShipoutPictureBG*{%
%  \AtPageUpperCenter{%
%    \raisebox{-1.5cm}{%
%      \makebox[0pt]{\small
%      To appear in ACM Int'l Conf. on Hybrid Systems: Computation and Control (HSCC), and ACM/IEEE Int'l Conference on Cyber-Physical Systems (ICCPS) 2026}%
%    }%
%  }%
%}

\AddToShipoutPictureBG*{%
  \AtPageUpperLeft{%
    \raisebox{-1.5cm}{%
      \makebox[\paperwidth][c]{\small
      To appear in ACM Int'l Conf. on Hybrid Systems: Computation and Control (HSCC), and ACM/IEEE Int'l Conference on Cyber-Physical Systems (ICCPS) 2026}%
    }%
  }%
}

%\AddToShipoutPictureBG*{%
%  \AtPageUpperLeft{%
%    \hspace*{\paperwidth}%
%    \raisebox{-1.5cm}[0pt][0pt]{%
%      \makebox[0pt][r]{\small
%      To appear in ACM Int'l Conf. on Hybrid Systems: Computation and Control (HSCC), and ACM/IEEE Int'l Conference on Cyber-Physical Systems (ICCPS) 2026}%
%    }%
%  }%
%}
%\begin{center}
%{\Large\bfseries My Header Title}\\[0.5em]
%{\small Optional subtitle or info}
%\end{center}
%
%\vspace{1em}
      
\title{
    %Data-Driven Modeling of Non-Linear Systems for Uncertainty Propagation and Probabilistic Evaluation
    Learning Nonlinear Continuous-Time Systems for Formal Uncertainty Propagation and Probabilistic Evaluation
}

\author{Peter Amorese}
\email{Peter.Amorese@colorado.edu}
\affiliation{%
    \institution{University of Colorado Boulder}
    \city{Boulder}
    \state{Colorado}
    \country{USA}
}
\author{Morteza Lahijanian}
\email{Morteza.Lahijanian@colorado.edu}
%\affiliation{\institution{University of Colorado Boulder}}
\affiliation{%
    \institution{University of Colorado Boulder}
    \city{Boulder}
    \state{Colorado}
    \country{USA}
}

%\thispagestyle{empty}
%\pagestyle{empty}

%%%%%%%%%%%%%%%%%%%%%%%%%%%%%%%%%%%%%%%%%%%%%%%%%%%%%%%%%%%%%%%%%%%%%%%%%%%%%%%%
\begin{abstract}
% Non-linear ordinary differential equations (ODE) can effectively model many real world dynamical systems. 
% Often, when using such models to predict future outcomes, accounting for uncertainty in the initial state can enhance the quality and robustness of the model's predictions. 
% Unfortunately, propagating uncertainty through non-linear ODEs is mathematically challenging, often requiring sampling-based approximation. In scenarios where the ODE is unknown and must be learned from data, a central question arises: can one choose a model architecture that readily permits formal uncertainty prediction and quantification of probabilistic events? This paper presents a novel continuum dynamics perspective on modeling for feasible uncertainty propagation based on constructing formal Taylor series approximations of probabilistic events. We provide a rigorous theoretical analysis of i) sufficient conditions for soundness of the approach and ii) universal asymptotic convergence of the model and its predictions. We validate the proposed approach with empirical case studies, emphasizing the efficacy of the proposed method particularly for prediction of rare-events.
%, particularly when quantifying the probability of the system entering a region at a future time.

Nonlinear ordinary differential equations (ODEs) are powerful tools for modeling real-world dynamical systems. However, propagating initial state uncertainty through nonlinear dynamics, especially when the ODE is unknown and learned from data, remains a major challenge. This paper introduces a novel continuum dynamics perspective for model learning that enables formal uncertainty propagation by constructing Taylor series approximations of probabilistic events. We establish sufficient conditions for the soundness of the approach and prove its asymptotic convergence. Empirical results demonstrate the framework’s effectiveness, particularly when predicting rare events.
\end{abstract}


\maketitle
%%%%%%%%%%%%%%%%%%%%%%%%%%%%%%%%%%%%%%%%%%%%%%%%%%%%%%%%%%%%%%%%%%%%%%%%%%%%%%%%

\section{Introduction}
Modeling real world dynamical systems is a core tenant of predicting their future behavior. 
Specifically, nonlinear ordinary differential equations (ODE) are known to capture the dynamics of many continuous-space, continuous-time systems~\cite{ehrendorfer1994liouville, sun2019nonlinear}.
When the ODE is not known, machine learning methods can generate accurate estimates from data.
Often, in such situations, there may also be uncertainty over the initial conditions of the model, which may lead to drastically different predictions.
However, using such models for predicting the future behavior of the system, particularly subject to uncertainty in the initial condition, is challenging since nonlinear ODEs generally do not admit analytical solutions.
This paper seeks to understand the connection between the choice of model and the prediction under uncertainty. In particular, \textit{can we learn a non-linear ODE model such that formal uncertainty propagation becomes more feasible?} 
% We present a continuum-dynamics perspective on the uncertainty propagation and modeling problem, and rigorously study an uncertainty propagation approach based on Taylor series approximation. 
We aim to address this question from a continuum-dynamics perspective on the uncertainty propagation and modeling problem.
% , and rigorously study an uncertainty propagation approach based on Taylor series approximation.

% This work focuses 

Uncertainty propagation is usually performed for the purposes of probabilistic evaluation, i.e., predicting the probability that the system's state $\px$ is in a certain region $R$ at a future time of interest $\toi$.
This task combines two mathematically challenging procedures: (i) propagating a probability density function (pdf) through the nonlinear model, and then (ii) integrating the resulting pdf over $R$. Directly propagating the pdf inherits the same analytical infeasibility as solving the nonlinear ODE, and thus, generally requires approximation. Furthermore, even supposing a tractable expression for the density, integration over $R$ may also be infeasible, requiring another approximation technique.

Density propagation is a well-studied problem, governed by the Liouville partial differential equation \cite{brockett2007optimal, li2008principle}.
Physics-informed methods learn an approximate solution to the Liouville equation as a neural-network or generative model using physics-informed learning \cite{meng2022learning, liu2022neural, he2025adaptive, kong2024error}. For probabilistic evaluation however, candidate models of the pdf, e.g., a neural network, are generally not analytically integrable, and thus cannot be easily used for probabilistic evaluation.
For quantification of probabilistic events, Gaussian mixture models combined with linear approximations of the dynamics have been used to reconstruct integrable propagated densities \cite{terejanu2008uncertainty}, but they struggle to provide formal bounds on probabilistic events. Moment-based polynomial optimization techniques \cite{streif2014probabilistic, covella2022uncertainty} can provide formal bounds on probabilistic events for polynomial ODEs, while being limited to bounded-support initial distributions. Discretization-based reachability approaches use overlapping forward reachable sets to estimate the probability \cite{gray2024verified}; however, such methods struggle to provide convergence guarantees due the strict over-approximations of reachable sets.

Sampling-based particle simulations are an effective approximate method of uncertainty propagation through non-linear ODEs \cite{yang2018effectiveness}.
Unlike density propagation methods, the propagation and probabilistic evaluation problem is very straight forward to approximate via sampling-based approaches. Initial states can be sampled, and each sample can be simulated forward with numerical integration. The ratio of samples that end in $R$ at time $\toi$ to the total number of samples estimates the desired probability. 
% While they provide simple computations, sampling-based approaches fail to yield a formal hard bound on the predicted probability.  Hence, they must instead rely on confidence-based concentration inequalities. Furthermore, 
Despite their simplicity, these methods do not provide formal, hard bounds on the resulting probability estimates and thus rely on confidence-based concentration inequalities instead. Moreover,
the accuracy of such methods degrades heavily when predicting rare events \cite{chan2011rare}, i.e., probabilities on the order of $10^{-6}$ and below. 

With regard to modeling, i.e., learning the ODE, SINDy \cite{brunton2016discovering} is a prominent framework for efficiently recovering the system's dynamics from data. Deep learning approaches can also learn the ODE as an expressive neural network and have straight-forward applications for (approximate) density propagation via continuous normalizing flows \cite{chen2018neural}. While these methods can learn highly accurate models, they do not take into account formal uncertainty propagation and evaluation and hence generally must rely on sampling-based methods. Recent work \cite{amorese2025universal} has tackled the joint modeling and uncertainty propagation problem; however, the methods only apply to discrete-time systems. 

This work addresses the challenges of propagating the pdf by introducing a continuum dynamics-based simplification to the propagation and evaluation problems, and by coupling the learned ODE architecture to match the requirements of the propagation procedure. 
%We view the problem as tracking a \textit{control mass} that containing the desired probability mass and moves subject to a velocity field. 
With a clever domain transformation, we can reduce the complicated state-prediction and integration problem to approximation of a scalar function, capturing the event-probability throughout time. 
A Taylor series approximation of the probability function is constructed to formally bound the model's probabilistic predictions.
We provide in-depth analysis of the restrictions on the learned model's architecture, and suggest a learning framework that possess strong asymptotic convergence guarantees.
%We approximate and upper-bound this function using Taylor series, and refine this approximation using off-the-shelf reachable set algorithms. 

The main contributions of this work are four-fold:
\begin{enumerate}[label=\roman*.]
    \item a novel continuum-dynamics based approach to uncertainty propagation,
    \item a rigorous analysis of the soundness and modeling constraints of the approach,
    \item an algorithm that uses off-the-shelf reachability tools to improve the predictions for longer horizons, and
    \item experimental results and analysis of the method.
\end{enumerate}

\subsection{Notation}
Vectors are denoted in bold, e.g. $\px = [x_0, \ldots, x_n]$.
For a vector $y \in \reals^n$, we denote its $i$-th dimension by $y_i$.
The $n$-length vector of ones is denoted $\mathbf{1}_n$.
Let $\conts{U}$ denote the set of all real-valued analytic functions on a set $U$.

\section{Problem Formulation}

Consider a dynamical system described by
\begin{equation} \label{eq:dynamics}
    \dot{\px} = f(\px)
\end{equation}
where $\px \in \reals^n$, and $f: \reals^n \to \reals^n$ is the vector field and possibly nonlinear. We assume that $f$ is globally Lipschitz continuous and each $f_i \in \conts{\reals^n}$.
%each component $f_i$ is $m$-times differentiable on $\reals^n$, denoted $f_i\in \conts{m}{\reals}$. 
Starting from an initial state $\px_0$, the system follows a unique trajectory flow-function $\phi(\px_0, t)$ where $t \in \reals_{\geq 0}$ and $\phi$ satisfies 
the ordinary differential equation
\begin{equation} \label{eq:ode}
    \frac{\partial \phi}{\partial t}(\px_0, t) = f(\phi(\px_0, t)).
\end{equation}

In this paper, we focus on the setting where $f$ is \emph{unknown}, and the initial state $\px_0$ has \emph{uncertainty}, characterized by the 
% random variable $\px_0 \sim p(\px_0)$ with 
distribution $p_0(\px_0)$.  Specifically, we assume that $p_0(\px_0)$ is a Normal distribution with mean $\mu$ and (non-degenerate) covariance matrix $\Sigma$, i.e., $\px_0 \sim p_0(\px_0) = \mathcal{N}(\mu, \Sigma)$, and that $\Sigma$ is diagonal
without loss of generality\footnote{Using a (linear) Mahalanobis transformation, System~\eqref{eq:dynamics} can always be mapped to a space, where the covariance of the initial state is diagonal.}.

Since $f$ is unknown, we assume a dataset of the form $\mathcal{D} = \{(\hat{\px}_i, f(\hat{\px}_i))\}_{i\in \mathbb{N}}$ is available, which allows learning an approximation of $f$. Let $\hat{f}$ denote the learned vector field and $\hat{\phi}(\px, t)$ be its corresponding flow function solution.

% In this paper, we consider the problem where $f$ is unknown, and must be learned from the data $\mathcal D = \{(\hat{\px}, f(\hat{\px}))_j\}$ where $j \in \{1, \ldots, p\}$. Let $\hat{f}$ denote the learned model and $\hat{\phi}(\px, t)$ denote the flow function that solves $\hat{f}$. The system's initial state has uncertainty, captured by the random variable $\px_0 \sim p(\px_0)$. For this work, we assume the initial state components are independent and normally distributed, i.e. $p(\px_0) = \mathcal N(\mu, \Sigma)$ where $\Sigma$ is a diagonal covariance, and thus can be factored as $p(\px_0) = p(x_{0, 1}) \cdots p(x_{0, n})$. 

Our overarching goal is to determine (a bound on) the probability of System~\eqref{eq:dynamics} being in a given region of interest $R \subset \mathbb{R}^n$ at a specific time $t=\toi$, denoted as $P(\phi(\px_0, \toi) \in R)$. 
This generally involves two steps: learning the model, and propagating uncertainty through the learned model.
This paper mainly focuses on the latter, the sub-problem of computing a formally guaranteed probability (bound) with respect to the \textit{learned} model, i.e., $P(\hat{\phi}(\px_0, \toi) \in R)$. Consequently, by ensuring that the learned model $\hat{f}$ possesses universal approximation and asymptotic convergence properties, the model's predictions thereby inherit the same convergence guarantees. Formally, we can define the convergence properties as follows.

\begin{definition}[Convergent Universal Estimator] \label{def:convergent}
    Let $\mathcal F_\theta \subset \conts{\uspace^n}$ be a class of parameterized functions 
    and $\mathcal L_{\mathcal D} : \mathcal F_\theta \rightarrow \reals$ be a loss function that is convex in the parameters $\theta$. 
    Let $\hat{f}_\theta^\star \in \mathcal F_\theta$ denote the global minimizer of the loss, i.e.
    \begin{equation}
        \hat{f}_\theta^\star = \argmin_{\hat{f}_\theta \in \mathcal F_\theta} \mathcal L_{\mathcal D}(\hat{f}_\theta).
    \end{equation}
    Then, for some arbitrary closed hyperrectangle $U \subset \uspace^n$, $\hat{f}^\star_\theta$ is a \emph{convergent universal estimator} if 
    \begin{subequations}
        \begin{align}
            \hat{f}_\theta^\star = \argmin_{\hat{f}_\theta \in \mathcal F_\theta} \mathcal L_{\mathcal D}(\hat{f}_\theta) 
        \end{align}
        and
        \begin{align}
            \lim_{|\theta|, |\mathcal D| \rightarrow \infty} \|\hat{f}_\theta^\star - f\|_{\mathcal C^k} = 0.
        \end{align}
    \end{subequations}
    for any $k > 0$ where $\| \cdot \|_{\mathcal C^k}$ denotes the $\mathcal C^k$-norm.
\end{definition}

%\begin{definition} [Convex Least-Squares Universal Approximator] \label{def:universal}
%    Let $\chi \subset \reals^n$ be a compact bounded domain and let $\conts{\chi}$ be the set of analytic functions on $\chi$. Let $\mathcal{L}(\hat{f}; f) = \mathbb{E}[(\hat{f}(\px) - f(\px))^2]$ \pa{check this} be the least-squares loss function.
%    A convex universal approximator is a class of parametric functions $F_\theta \subseteq \conts{\chi}$ parameterized by $\theta$ that satisfies the two following properties:
%    \begin{enumerate}
%        \item for any $\epsilon > 0$, there exists a $\hat{f}_\theta \in F_\theta$ such that
%        \begin{equation}
%            \sup_{\px \in \chi} \|f(\px) - \hat{f}_\theta(\px)\| \leq \epsilon
%        \end{equation}
%        for any $f(\px) \in \conts{\chi}$, and
%        \item $\mathcal{L}(\hat{f}_\theta; f)$ is convex in $\theta$.
%    \end{enumerate}
%\end{definition}
%\noindent
%Definition \ref{def:universal} characterizes the strong asymptotic convergence properties necessary for the learned model. Namely, $F_\theta$ must parameterize a class of universal approximators for which the loss function is convex in the parameters, allowing the training procedure to actually obtain the least-squares minimizer.

Given a learned model $\hat{f}_\theta$, we are interested in propagating the uncertainty in $\px_0$ to predict the probability of the state entering a region of interest $R$ at a given time $t=\toi$, i.e. 
\begin{equation} \label{eq:prob_int}
    P(\hat{\phi}(\px_0, \toi) \in R) = \int_R p(\px_{\toi}) d\px_\toi
\end{equation}
where $p(\px_\toi)$ is the system's state distribution at time $t=\toi$ computed via
\begin{equation} \label{eq:prop}
    \log p(\px_\toi) = \log p_0(\px_0) - \int_{t=0}^{t=\toi} \nabla \cdot \hat{f}_\theta(\px_t) dt.
\end{equation}
Since $\hat{f}_\theta$ is non-linear,
predicting $P(\hat{\phi}(\px_0, \toi) \in R)$ using \eqref{eq:prob_int} and \eqref{eq:prop} is subject to two (generally) intractable integrals.
%, requiring approximation.
%Prediction error with respect to the true system $f$ is subject to two sources of error: i) \textit{learning}-error with respect to the true model and ii) \textit{approximation}-error with respect to the learned model. 

This work focuses on the coupling between a universal model, and a method of computing convergent formal bounds on \eqref{eq:prob_int}.
Formally, we consider the two-part problem as stated below.
\begin{problem} \label{prob}
    Consider System~\eqref{eq:dynamics} with unknown $f$ and uncertain initial state $\px_0 \sim \mathcal{N}(\mu,\Sigma).$
    Given data $\mathcal D$, a hyperrectangular region of interest $R \subset \reals^n$, and a time of interest $\toi$, 
    \begin{enumerate}%[label=(\arabic*)]
        \item learn convergent universal estimator model $\hat{f}_\theta$, and \label{prob:learn}
        \item compute $\bar{P}_\toi \in [0,1]$ such that  
        $P(\hat{\phi}(\px_0, \toi) \in R) \leq \bar{P}_\toi$. \label{prob:prop}
    \end{enumerate}
    
    %\begin{itemize}
    %    \item $P(\hat{\phi}(\px_0, \toi) \in R) \leq \bar{P}$, and
    %    \item $\bar{P}_m$ converges to the \textit{true} probability, i.e.,
    %    \begin{equation}
    %        \lim_{(|\mathcal{D}|, |\theta|, m) \rightarrow \infty} |\bar{P}_m - P(\phi(\px_0, \toi) \in R)| = 0.
    %    \end{equation}
    %\end{itemize}
\end{problem}

Initially, the \textit{model-learning} Problem~1.\ref{prob:learn} and the \textit{uncertainty propagation} Problem~1.\ref{prob:prop} may appear as disjoined problems.
%For example, if one were to directly regress $\hat{f}$, then the probability distribution $p(\px_t)$ can be propagated according to
%For general non-linear $f$, \eqref{eq:prop} is analytically intractable. Furthermore, even with an expression for $p(\px_\toi)$, the integral over $R$ in \eqref{eq:prob_int} is also generally intractable. 
%This makes computing a probabilistic bound $\bar{P}$ is a very challenging problem. 
In this work, we study the formal uncertainty propagation problem to reveal a structure of $\hat{f}_\theta$ that allows for tractable computation of $\bar{P}_\toi$ without compromising universal representational power. For ease of presentation, we drop the subscript $\theta$ and assume that $\hat{f}$ is a parameterized estimator.

\begin{figure*}[t]
    \centering

    %==================== TOP ROW ====================
    \begin{subfigure}{\textwidth}
        \centering
        \begin{tabular}{ccccc}
            % Row of text labels (only first and last non-empty)
            \makebox[0.18\textwidth][c]{\textbf{$t=0$}} &
            \makebox[0.18\textwidth][c]{} &
            \makebox[0.18\textwidth][c]{} &
            \makebox[0.18\textwidth][c]{} &
            \makebox[0.18\textwidth][c]{\textbf{$t=\toi$}} \\
            [0.0em]
            % Row of images
            \includegraphics[width=0.18\textwidth]{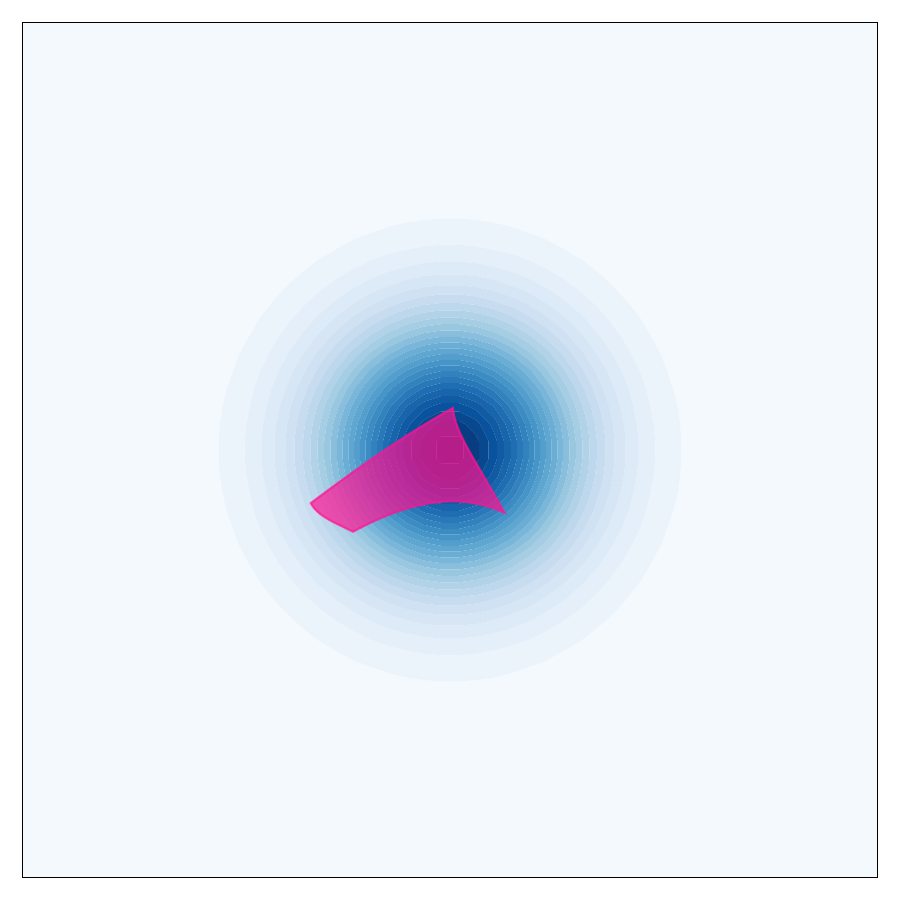} &
            \includegraphics[width=0.18\textwidth]{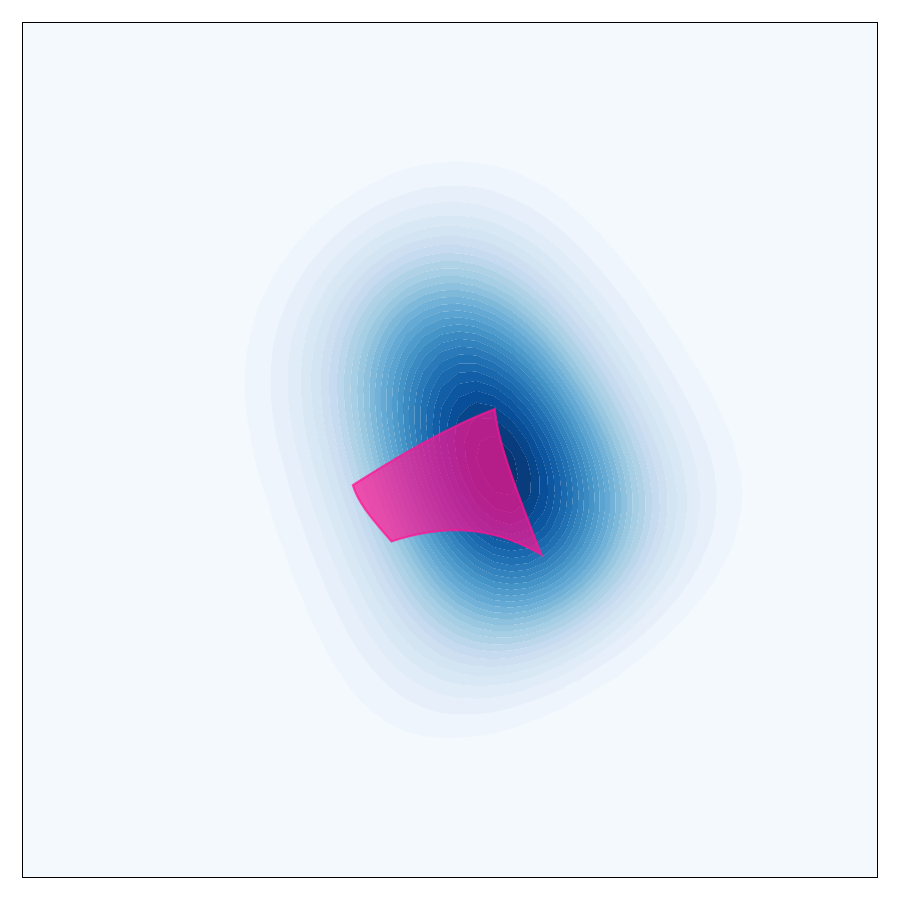} &
            \includegraphics[width=0.18\textwidth]{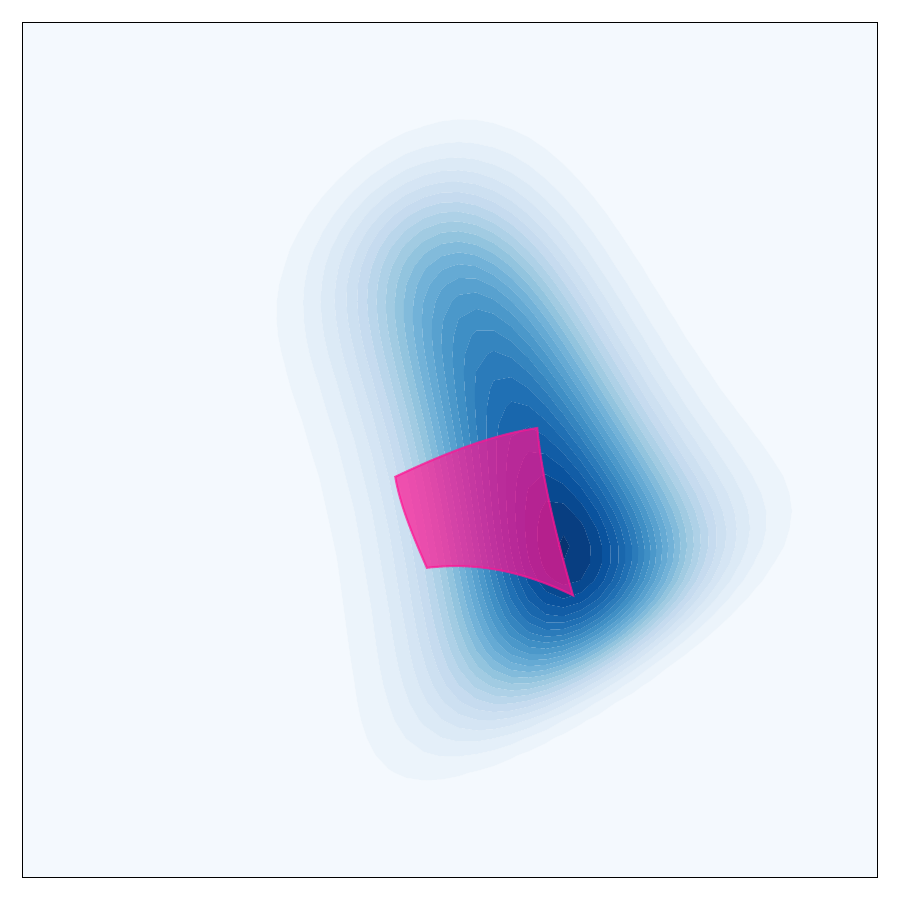} &
            \includegraphics[width=0.18\textwidth]{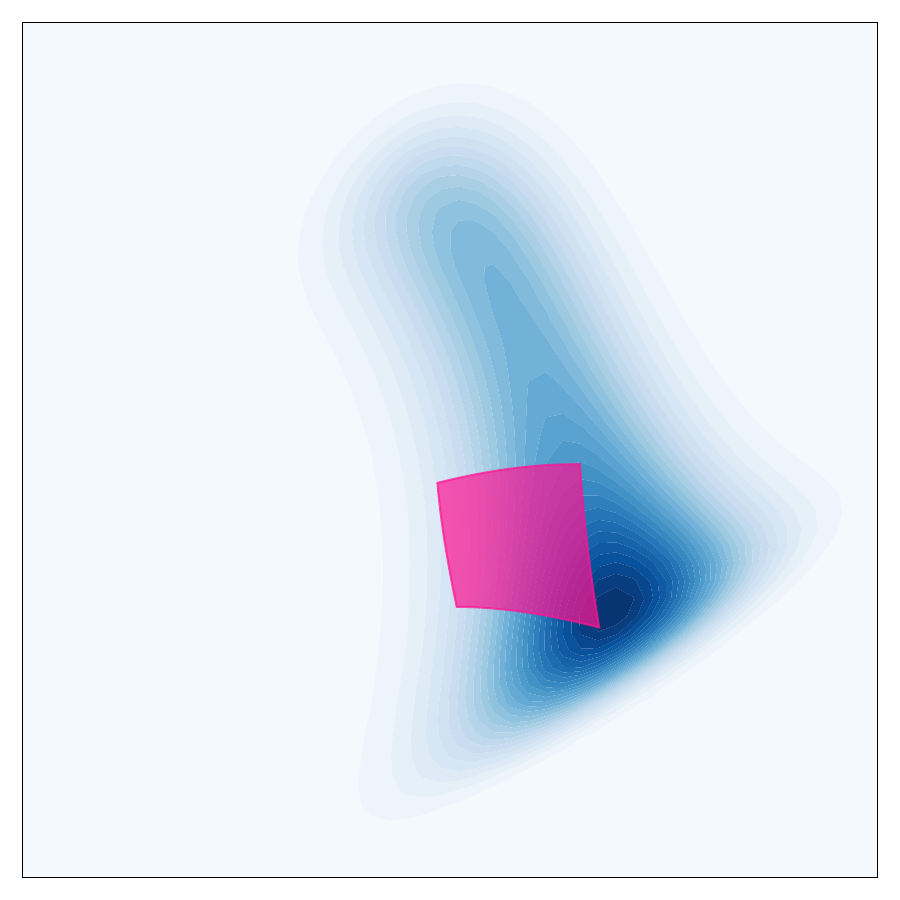} &
            \includegraphics[width=0.18\textwidth]{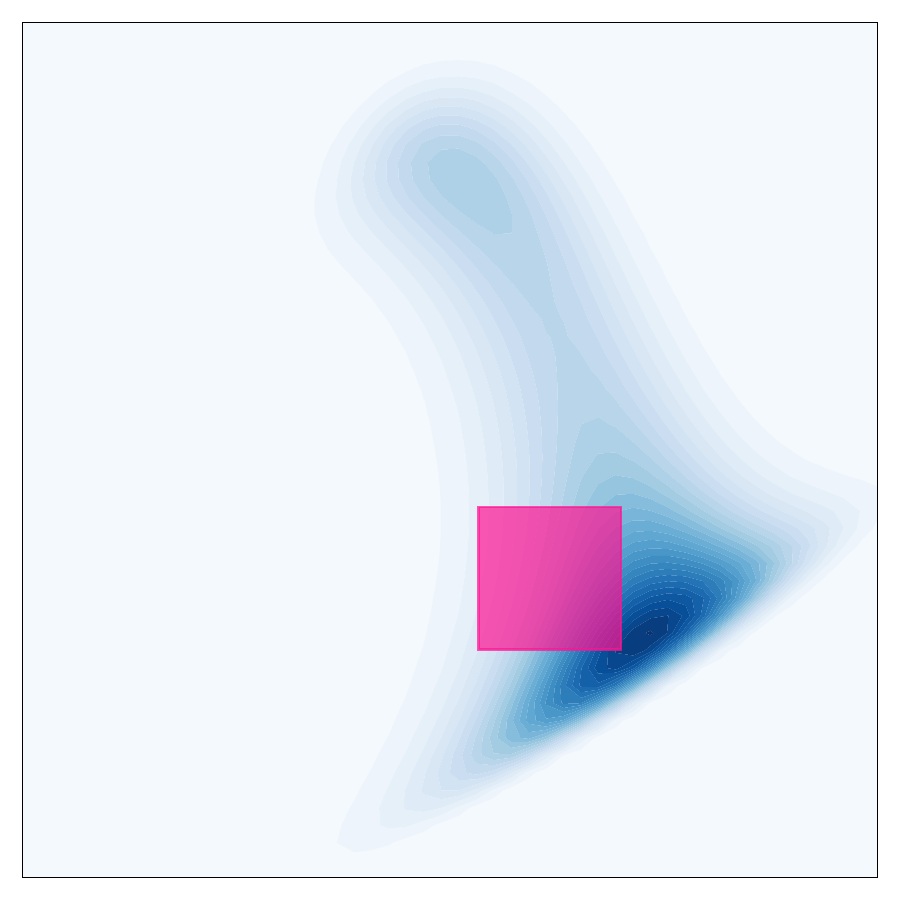}
        \end{tabular}
        \caption{State-space ($\reals^n$) evolution}
        \label{fig:ss}
    \end{subfigure}

    \vspace{1em}

    %==================== BOTTOM ROW ====================
    \begin{subfigure}{\textwidth}
        \centering
        \begin{tabular}{ccccc}
            % Row of text labels (only first and last non-empty)
            \makebox[0.18\textwidth][c]{\textbf{$t=0$}} &
            \makebox[0.18\textwidth][c]{} &
            \makebox[0.18\textwidth][c]{} &
            \makebox[0.18\textwidth][c]{} &
            \makebox[0.18\textwidth][c]{\textbf{$t=\toi$}} \\
            [0.0em]
            % Row of images
            \includegraphics[width=0.18\textwidth]{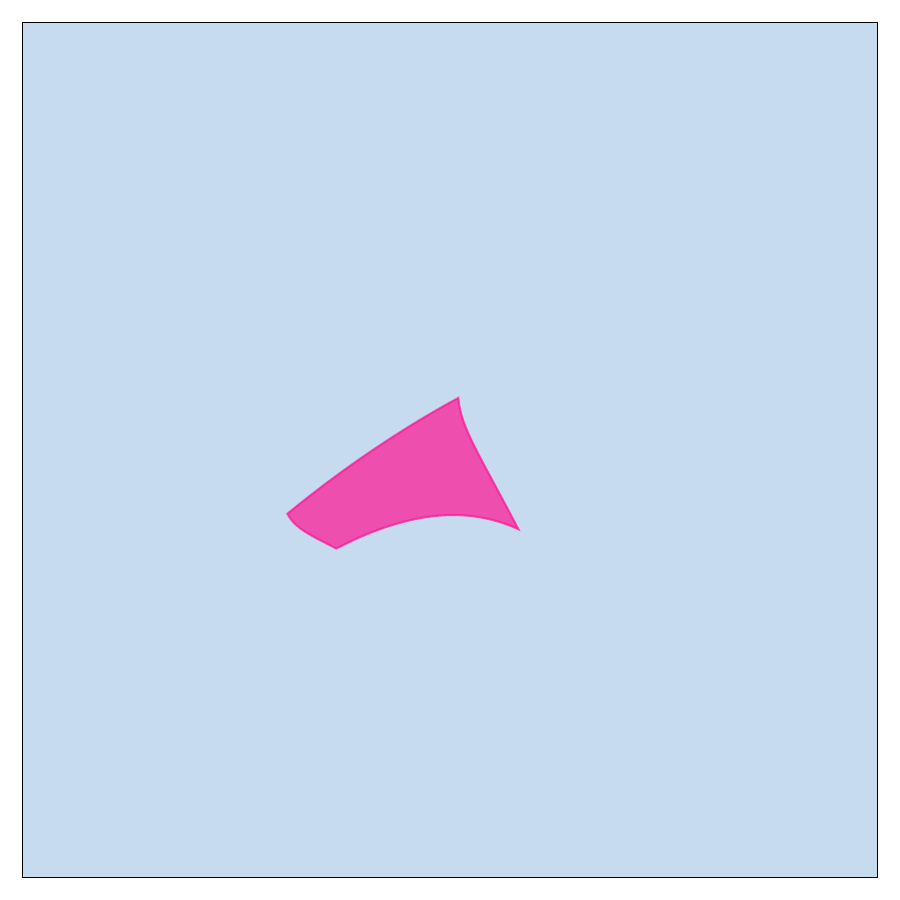} &
            \includegraphics[width=0.18\textwidth]{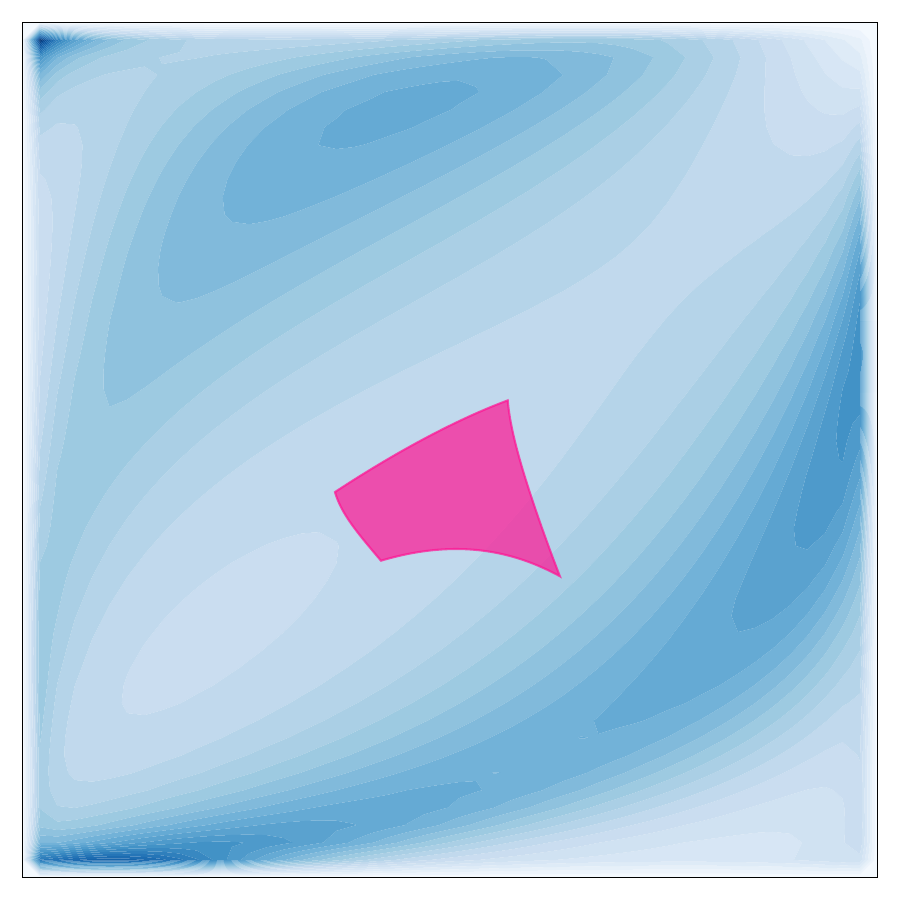} &
            \includegraphics[width=0.18\textwidth]{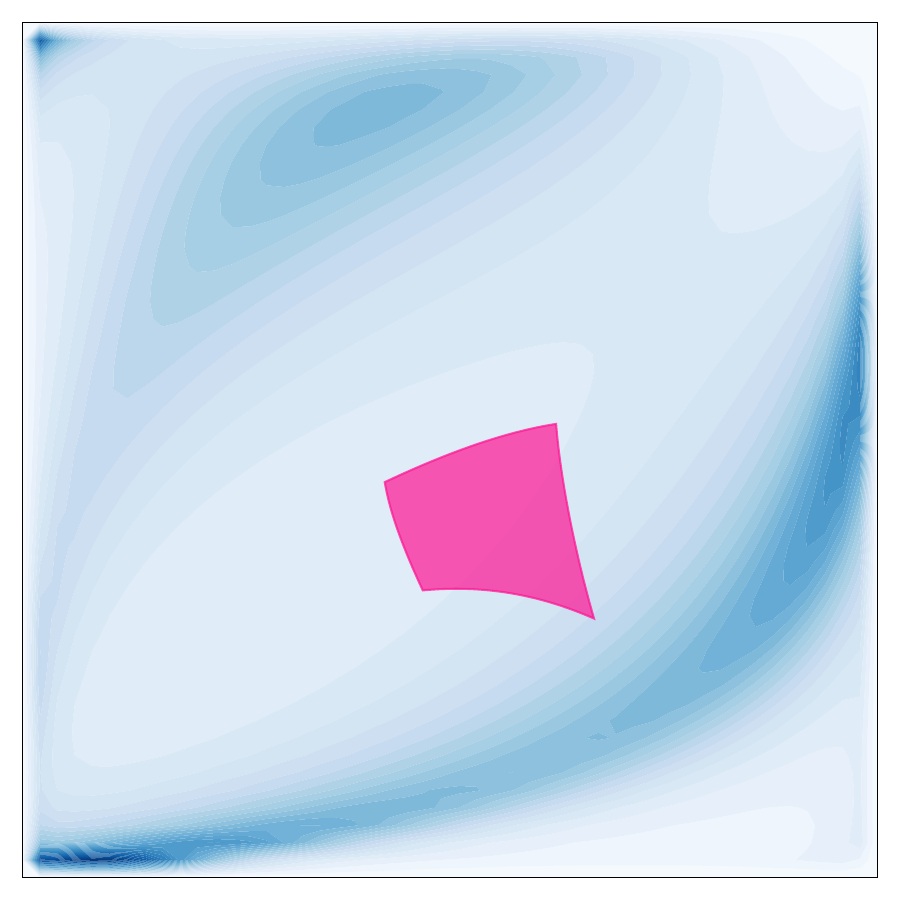} &
            \includegraphics[width=0.18\textwidth]{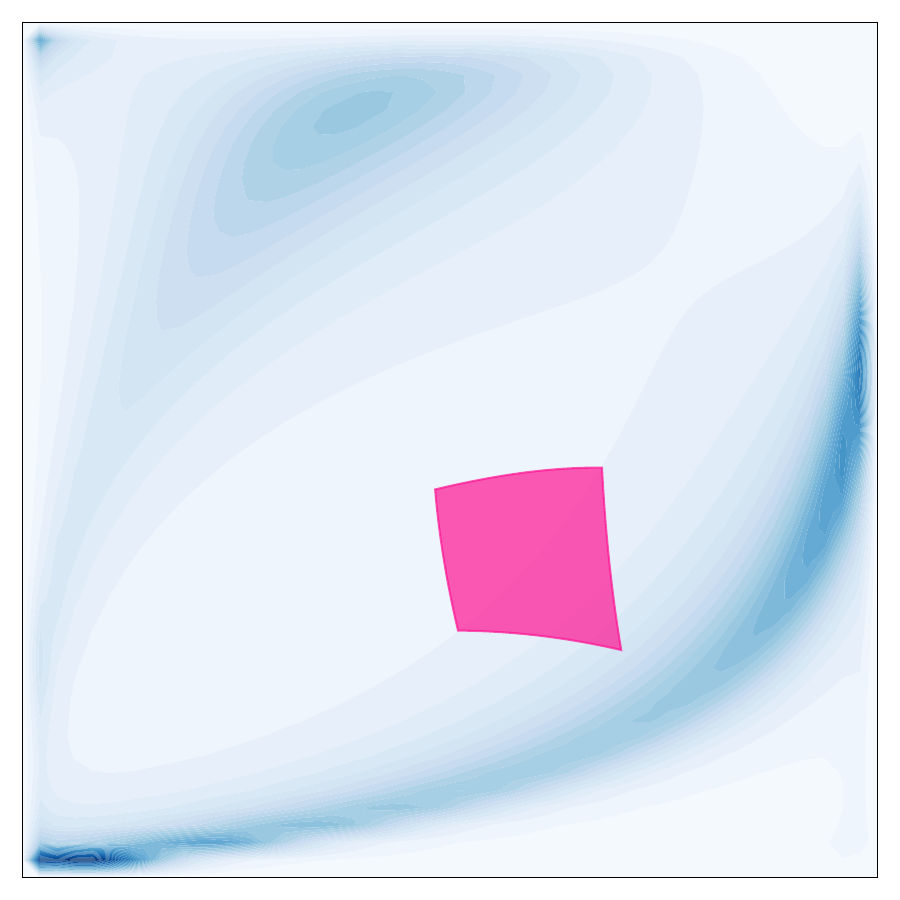} &
            \includegraphics[width=0.18\textwidth]{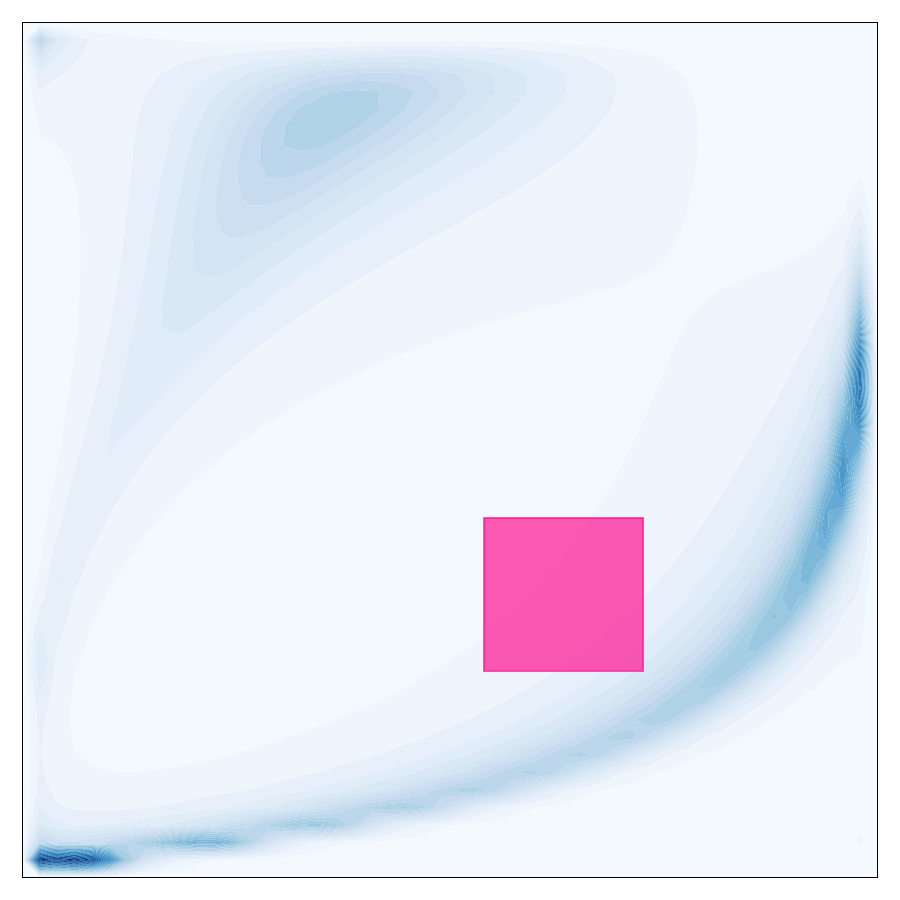}
        \end{tabular}
        \caption{Transformed-space ($\uspace^n$) evolution}
        \label{fig:us}
    \end{subfigure}

    \caption{ Example evolution of a PDF subject to a non-linear $\hat{f}$. We wish to integrate $p(\px_\toi)$ (blue density) at time $t=\toi$ over $\cmass(\toi)=R$ (pink region). At time $t=0$, the density is simple, i.e., Gaussian in $\reals^n$ or uniform in $\uspace^n$, but becomes complex as time progresses. At time $t=\toi$, $R$ is a simple rectangle, and the shape becomes more complex as time regresses. Since probability mass is a conserved quantity, the volume and density of $\cmass$ fluctuate such that its probability mass is constant throughout time.}
    \label{fig:ex-deformation}
\end{figure*}

\section{Uncertainty Propagation} \label{sec:prop}
This section proposes a method for computing probability bound $\bar{P}_\toi$.
Informed by the requirements of this procedure, we can choose a regressor $\hat{f}$ that enables tractable computation of $\bar{P}_\toi$.

\subsection{Control Mass Perspective}
The uncertainty propagation and integration problem can be alternatively viewed as tracking the properties of a \textit{control mass}, i.e., a collection of a fixed set of infinitesimal \textit{material elements} (e.g., states) that move subject to $\hat{f}$
\cite{gurtin1982introduction}. A non-linear $\hat{f}$ causes the volume of each of the material elements to change in time, and thus, distorting the shape of the control mass.

Fig. \ref{fig:ex-deformation} illustrates the relevant continuum dynamics of the control mass with respect to the evolution of the density in time. We study the control mass $\cmass$ of particles which occupy region $R$ at time $\toi$ (pink box in the rightmost plots), which, since $\hat{\phi}$ is a one-parameter group \cite{zeidler2012applied}, can be readily defined as
\begin{equation}
    \cmass(t) = \{\px_t \mid \hat{\phi}(\px_t, \toi-t) \in R\}.
\end{equation}
Determining the total \textit{mass} of $\cmass$ is equivalent to the computation in \eqref{eq:prob_int}. In Fig. \ref{fig:ss}, the state-space density is initially Gaussian, and deforms into a complex distribution subject to \eqref{eq:prop}. 
Conversely, at $t=0$, the $\cmass(0)$ begins with a highly complex shape, such that at $t=\toi$, the particles occupy the simple rectangular shape of $R$. Given a fixed $R$, the control mass is uniquely defined by the time instance $\toi$, since the set of material elements (or initial states) that occupy $R$ at $t=\toi_1$ may be different from those that occupy $R$ at a different time $\toi_2 \neq \toi_1$.

With this perspective, there are two equivalent means of computing $P(\hat{\phi}(\px_0, \toi) \in R)$:
(i) integrating the \textit{complicated} deformed density over a \textit{simple} rectangular region of integration at $t=\toi$: 
\begin{equation} \label{eq:future_time_int}
    \int_{R=\cmass(\toi)} p(\px_\toi) d\px_\toi
\end{equation}
or (ii) integrating a \textit{simple} density over a \textit{complicated} region of integration $t=0$:
\begin{equation} \label{eq:init_time_int}
    \int_{\cmass(0)} p_0(\px_0) d\px_0.
\end{equation}
At first glance, either option seems equally challenging. To elucidate the benefits of option (ii), we employ a domain transformation to make the initial distribution \textit{uniform}, effectively reducing the complicated integration in \eqref{eq:init_time_int} to calculating the volume of $\cmass(0)$.

Let $\uspace^n = (0, 1)^n$ denote the open unit-hyperrectangle.
Consider the diffeomorphism $\psi : R^n \rightarrow \uspace^n$ defined with respect to the parameters $\mu$ and $\Sigma$ of $p_0(\px_0)$:
\begin{equation}
    \psi_i(\px) = c(x_i; \mu_i, \Sigma_{ii})
\end{equation}
where $c(\cdot; \mu_i, \Sigma_{ii})$ is the cumulative distribution function (cdf) of $p_0(x_{0,i})$. 
Let $\pu = \psi(\px)$ denote the resulting transformed state in $\uspace^n$.
Since $\psi$ is constructed to be the cdf of each initial state component, the resulting initial distribution in $\uspace^n$
is uniform, i.e., $p_0(\pu_0) = \mathcal U((0, 1)^n)$. Consequently, \eqref{eq:init_time_int} simplifies to
\begin{align} \label{eq:vol}
    \int_{\cmass(0)} p_0(\px_0) d\px_0 = \int_{\psi(\cmass(0))} 1 d\pu_0 
    = \vol{\psi(\cmass(0))}.
\end{align}
Furthermore, using the following chain-rule relationship: 
\begin{align}
    \dot{\pu} &= \frac{\partial\psi}{\partial \px} \frac{\partial \px}{\partial t} \notag \\
    &= p_0\big(\psi^{-1}(\pu)\big) \dot{\px}\big(\psi^{-1}(\pu)\big),
\end{align}
any state-space vector field model can be equivalently expressed in $\uspace^n$. However, doing so imposes boundary conditions on the $\pu$-space model. Specifically, 
\begin{equation} \label{eq:boundary_conds}
     \dot{\pu}_i = 0 \quad \text{ if }  \quad u_i = 0 \text{ or } u_i = 1.
\end{equation}
 This ensures that the system never leaves $\uspace^n$ (and thus never leaves the state-space), since there must always exist a bijection between $\reals^n$ and $\uspace^n$.
The hyperrectangular region of interest $R$ can also be mapped to $\uspace^n$, i.e., $R_u = \psi(R)$, yielding another hyperrectangular region, since $\psi$ is fully decoupled and monotone.

Fig. \ref{fig:us} illustrates the equivalent propagation scenario in the transformed space $\uspace^n$. 
In both $\reals^n$ and $\uspace^n$, the probability mass of $\cmass$ is \textit{conserved} throughout time. In $\uspace^n$, however, the density is uniform at $t=0$, and thus, knowing \textit{only} the volume of $\cmass$ is sufficient to calculate the mass. Therefore, performing propagation in $\uspace^n$ elliminates the need for (i) performing integration over the initial distribution, (ii) tracking the $\pu$-space density, and (iii) tracking the shape of $\cmass(t)$. 
Both integrals involved with the computation in \eqref{eq:prob_int} reduce to tracking the evolution of the volume (a scalar quantity) of the $\pu$-space control mass $\psi(\cmass(t))$ backwards in time until $t=0$.
We henceforth study the evolution of the control mass evolution in $\uspace^n$, and, for simplicity of notation, we use $\hat{f}$, $\cmass$, and $R$ to denote the $\pu$-space vector field, $\pu$-space control mass, and $R_u$ respectively.

\subsection{Volume Function Taylor Expansion}
%Tracking only the volume of $\cmass_u$ over time entirely circumvents the need for maintaining a high-fidelity approximation of the shape. %We propose using a Taylor series expansion of the volume of $\cmass_u$ throughout time, denoted $\volf_{\cmass}(t)$.
%Suppose $\cmass(t)=R$. We wish to calculate the volume of $\cmass(0)$.
%The probability mass contained in $\cmass$ is simply the volume of $\cmass(0)$ since the density at $t=0$ is uniform. 

Let the function $\volf(\cdot; \toi) : \reals_{\geq 0} \rightarrow [0, 1]$ represent the volume of $\cmass(t)$ over time. The variable $t$ indicates \textit{where} in time a specific control mass is, whereas parameter $\toi$ indicates when the control mass occupies $R$, effectively capturing \textit{which} control mass is being tracked. For example, $\volf(t_1; \toi)$ and $\volf(t_2; \toi)$ describe the \textit{same} control mass (initial conditions) at different times, where as $\volf(t; \toi_1)$ and $\volf(t; \toi_2)$ describe \textit{different} control masses at the same time.
%the initial region $\cmass(0)$ such that, subject to $\hat{f}$, the region evolves to $\cmass(\toi) = R$. 
Leveraging the aforementioned properties of $\uspace^n$, the \textit{initial} volume equals the desired probability measure
\begin{equation}
    \volf(t=0; \toi) = P(\hat{\phi}(\pu_0, \toi) \in R).
\end{equation}
%From an alternate perspective, the control mass can be ``initialized'' at $t=\toi$ as $\cmass(\toi) = R$ and evolved \textit{backwards} in time, i.e., subject to $-\hat{f}$ until $t=0$.
%This can be computed by initializing $\cmass(t) = R$ and evolving subject to $-\hat{f}$, i.e., backward in time.
% Consider the function $\volf(t)$ which represents the volume of a control mass $\cmass(t)$ initialized at the future time $t$ as $\cmass(t)=R$ and evolves subject to $-\hat{f}$, i.e., backward in time. 
%Evaluating $\volf(t=T)$ is determines $\vol{\cmass(T)}$
%In other words, if a control mass occupying region $R$ is released in 
%In other words, $\volf(t)$ tracks the volume \textit{backward} in time of a control mass starts as $R$.
%, such that $\volf(T) = \vol{R}$.
Notice that $\volf(\toi; \toi) = \vol{R}$, and computing a $\bar{P}_\toi$ such that$\volf(0; \toi)\leq \bar{P}_\toi$ solves Problem~1.\ref{prob:prop}.
% condition (2) in Prob. \ref{prob}. 

Non-linear models $\hat{f}$ generate a volume function that generally cannot be expressed in closed form. 
%Additionally, since tracking the shape of $\cmass(t)$ is infeasible, we are restricted to geometric reasoning over only $R = \cmass(T)$. 
Moreover, the shape of $\cmass(t)$ has a simple description only at time $t=\toi$.
To address these challenges, we build a \textit{Taylor series} approximation of $\volf(t; \toi)$ about the point $t=\toi$:
\begin{equation} \label{eq:taylor}
    \volf(t; \toi) \approx \volfapprox^m(t; \toi) = \sum_k^m \frac{d^k \volf}{d t^k}\Big|_{t = \toi} \frac{(t-\toi)^k}{k!}.
\end{equation}
The Taylor series uses only point-wise information in the form of time derivatives of $\volf(t; \toi)$ at $t=\toi$ to approximate $\volf(t; \toi)$ elsewhere in time. By centering the Taylor expansion about $t=\toi$, the simplicity of the hyperrectangular shape of $\cmass$ enables exact computation of the expansion coefficients for certain classes of $\hat{f}$.

In order to calculate the coefficients, we use the continuum dynamics of $\cmass(t)$. 
%The volume of $\cmass(t)$ is a time-dependent scalar property of $\cmass(t)$. 
The time derivative of a (cumulative) time-dependent scalar quantity (such as the volume) over a material element $\cmass(t)$ is governed by Reynold's transport theorem \cite{gurtin1982introduction}, which states, for a given cumulative quantity of interest $\gamma(\pu, t)$,
\begin{equation} \label{eq:rtt}
    \frac{d}{dt} \int_{\cmass(t)} \gamma(\pu, t) d\pu = \int_{\cmass(t)} \frac{\partial \gamma}{\partial t} + \divergence \big(\gamma(\pu, t) \hat{f}(\pu)\big) d\pu.
\end{equation}
By letting $\gamma(\pu, t) = 1$, we can recover an expression for the first time derivative of the volume function:
\begin{equation} \label{eq:first_d}
    \frac{d}{dt}\volf(t; \toi) = \int_{\cmass(t)} \divergence \hat{f}(\pu) d\pu,
\end{equation}
which is a well-known property in continuum dynamics. Using this result, the second-order time derivative of the volume function is recovered by instead using $\gamma(\pu, t) = \divergence \hat{f}(\pu)$:
\begin{align} \label{eq:second_d}
    \frac{d^2}{dt^2}\volf(t;\toi) &= \frac{d}{dt} \int_{\cmass(t)} \divergence \hat{f} d\pu \\ \notag
    &= \int_{\cmass(t)} \divergence \big((\divergence \hat{f}) \hat{f} \big) d\pu.
\end{align}
One can apply the same substitution recursively to derive the expression for the $k$-th time derivative $\volf(t)$, as illustrated by the following lemma.

\begin{lemma} \label{lem:v_derivs}
    Let $\Gamma_{k=0}(\pu) = 1$ and $\Gamma_{k}(\pu)$ be recursively defined as 
    \begin{equation} \label{eq:gamma_operator}
        \Gamma_{k}(\pu) = \divergence \big(\Gamma_{k-1}\hat{f}\big).
    \end{equation}
    Then, the $k$-th time derivative of the volume function is
    \begin{equation} \label{eq:integ}
         \frac{d^k}{dt^k} \volf(t;\toi) = \int_{\cmass(t)} \Gamma_{k} d\pu.
    \end{equation}
\end{lemma}
\begin{proof}
    The proof follows from induction. For $k=1$, \eqref{eq:first_d} holds. Assume
    \begin{equation}
         \frac{d^k}{dt^k} \volf(t;\toi) = \int_{\cmass(t)} \Gamma_{k}(\pu) d\pu.
    \end{equation}
    Then,
    \begin{align}
        \frac{d^{k+1}}{dt^{k+1}} \volf(t;\toi) =& \frac{d}{dt} \int_{\cmass(t)} \Gamma_{k}(\pu) d\pu \notag \\
        =& \int_{\cmass(t)} \frac{\partial \Gamma_k}{\partial t} + \divergence \big(\Gamma_k(\pu) \hat{f}(\pu)\big) d\pu
    \end{align}
    and $\frac{\partial \Gamma}{\partial t}=0$ since $\Gamma_{k}(\pu)$ is not time dependent.
\end{proof}

Lemma \ref{lem:v_derivs} describes a recursive method of calculating the time derivatives of the volume function. The functions $\Gamma_k$ are related to the point-wise $k$-th order volumetric rates of change, which, when integrated over $\cmass(t)$, yields the cumulative rate of change of the control mass volume.
%By choosing to expand the Taylor series about $t=\toi$, the region of integration $\cmass(t)=R$ is hyperrectangular. 
The integral in \eqref{eq:integ} is over $\cmass(t)$, which, evaluated at the expansion time $t=\tau$, yields the hyperrectangle integration region $R$. Therefore, if the antiderivative of $\Gamma_k$ can be analytically determined, then the volume function derivatives can be computed exactly. As a consequence, the Taylor expansion approximation degrades \textit{only} with respect to time, and does not rely on any spatial approximation. Therefore, the prediction quality does not depend on where $R$ is, and the method can produce accurate estimates even if $R$ lies in the tail of the distribution, i.e. a \textit{rare-event} (see Sec. \ref{sec:experiments}).

\begin{remark}
    The volume function time derivatives can be computed with time-dependent vector field model $\hat{f}(\pu, t)$ by following the same recursive procedure.
\end{remark}

However, building a Taylor series expansion \eqref{eq:taylor} with arbitrary $\hat{f}$ may not actually converge to the true volume function. In fact, such convergence exists iff $\volf(t)$ is analytic. Since $\hat{f}$ implicitly generates $\volf(t)$, we must characterize which models $\hat{f}$ generate analytic volume functions. To do so, we can study how the operator \eqref{eq:gamma_operator} generates derivatives, and show that the derivatives cannot grow too quickly (causing the Taylor series to diverge). This process is very challenging for arbitrary analytic $\hat{f}$, however, by making assumptions about the analytic \textit{continuation} of $\hat{f}$ in the complex plane, a sub-factorial bound can be achieved, ensuring the Taylor series converges.

\begin{theorem} \label{thm:analytic}
    Let $\hat{g}(\pw) : \mathbb{C}^n \rightarrow \mathbb{C}^n$ be the analytic continuation of $\hat{f}(\pu)$ to the complex domain $\pw \in \mathbb{C}^n$ with $Re(\pw) \in \uspace^n$. 
    Define $B(r, \pw_0)$ as the $n$-dimensional complex polydisk $B(r, \pw_0) = \{ \pw \in \mathbb{C}^n \text{ s.t. } |w_i - w_{0, i}| < r \; \forall 1\leq i \leq n\}$. 
    Then, if $\hat{g}(\pw)$ is holomorphic on $B(r, \pw_0)$ with $r > 1$ and $\pw_0 = \frac{1}{2} \mathbf{1}_n + \mathbf{0}i$, then $\volf(t)$ is analytic with infinite radius of convergence.
\end{theorem}
\begin{proof}
    First, we recall sufficient conditions for the volume function to be analytic with infinite radius of convergence.
    Let $r^k(t) = \volf(t) - \volfapprox(t)$ be the remainder of a degree-$k$ Taylor expansion. By Taylor's theorem,
    \begin{equation} \label{eq:remainder_bound}
        |r^k(t)| \leq M_k \frac{(t-\toi)^{k+1}}{(k+1)!}
    \end{equation}
    where $\big|\frac{d^{k+1}}{dt^{k+1}}\volf(\xi)\big| \leq M_k$ for all $\xi \in [t, \toi]$.
    If $\volf(t)$ is analytic then the expansion must converge, i.e., $\lim_{k \rightarrow \infty} |r^k(t)| = 0$. Therefore, it suffices to show that
    \begin{equation}
        \lim_{k \rightarrow \infty} \frac{M_k}{(k+1)!} = 0,
    \end{equation}
    i.e., $M_k$ grows slower than the factorial. 

    Next, to obtain a sub-factorial bound on $M_k$, we leverage Cauchy's estimate for multivariate holomorphic functions \cite{hormander1973introduction}. Specifically, for a complex-valued function $g$ that is holomorphic on $B$,
    \begin{equation} \label{eq:cauchy_est}
        \Big| \frac{\partial g}{\partial w_i}(\pw) \Big| \leq |w_i - w_{0, i}|^{-1} \sup_{\pw \in B} |g|.
    \end{equation}
    We aim to use \eqref{eq:cauchy_est} to bound the growth of consecutive $\Gamma_k(\pw)$ generated using $\hat{g}$ instead of $\hat{f}$.
    Since the operator \eqref{eq:gamma_operator} consists of multiplication, addition, and differentiation, each $\Gamma_k(\pw)$ is composed of algebraic operations and thus must be holomorphic on $B$. Additionally, any function that is holomorphic on $B$ must be bounded on a poly disk with radius smaller than or equal to $r$.
    %i.e., there exists some $0 \leq A < \infty$ such that $|\sup_{\pw \in B(r=1, \pw_0)} |\hat{g}_u(\pw)| \leq A$.
    %Importantly, the operator \eqref{eq:gamma_operator} consists of multiplication and differentiation, therefore each $\Gamma_k(\pw)$ is holomorphic on $B$ and $\Gamma_k(\pw) = \Gamma_k(\pu)$ when $\pw = \pu$. 

    Substituting in $\hat{g}$, the operator \eqref{eq:gamma_operator} can be equivalently written as
    \begin{equation}
        \Gamma_{k}(\pw) = \sum_{i=1}^n \frac{\partial \Gamma_{k-1}}{\partial w_i}(\pw) \hat{g}_{i}(\pw) + \Gamma_{k-1}(\pw) \frac{\partial \hat{g}_{i}}{\partial w_i}(\pw) \label{eq:expanded_gamma_op}
    \end{equation}
    Let $A_{k}= \sup_{\pw \in B(1, \pw_0)} | \Gamma_{k}(\pw)|$, and $\epsilon = \max_{1 \leq i \leq n} |w_i - w_{0, i}|^{-1}$.
    Then, applying \eqref{eq:cauchy_est} to \eqref{eq:expanded_gamma_op} yields the recursion
    \begin{align}
        A_{k} \leq& \sum_{i=1}^n \epsilon A_{k-1} |\hat{g}_{u, i}(\pw)| + \epsilon A_{k-1} \Big|\frac{\partial \hat{g}_{i}}{\partial w_i}(\pw)\Big| \notag \\
        =& A_{k-1} \epsilon \Big(\sum_{i=1}^n |\hat{g}_{u, i}(\pw)| + \Big|\frac{\partial \hat{g}_{i}}{\partial w_i}(\pw)\Big|\Big) \label{eq:gamma_mag_ub}
        %=& A s(\hat{g}).
    \end{align}
    Observing that $s=\epsilon \sum_{i=1}^n |\hat{g}_{i}(\pw)| + \big|\frac{\partial \hat{g}_{i}}{\partial w_i}(\pw)\big|$ is independent of $k$, a non-recursive bound for $A_k$ is obtained as
    \begin{equation}
        A_k \leq s^k.
    \end{equation}
    Following from \eqref{eq:integ}, to obtain a bound on the $k$-th time derivative of $\volf(t)$ the integral can be bounded by multiplying an upper bound of the integrand with the largest possible region of integration $\cmass$ can occupy (all of $\uspace^n$), i.e.,
    \begin{align}
         \int_{\cmass(t)} \Gamma_{k} d\pu \leq& \Big(\sup_{t \in [0, \infty)} \vol{\cmass(t)} \Big) \sup_{\pw \in B(1, \pw_0)} | \Gamma_{k}(\pw)| \notag\\
         \leq& \vol{\uspace^n} A_k.
    \end{align}
    Hence, since $\vol{\uspace^n} = 1$, \ $\frac{d^k}{dt^k}\volf(t) \leq s^k$
    showing exponentially bounded (sub-factorial) growth for all $t$, sufficient for $\volf(t)$ to be analytic with infinite radius of convergence.
\end{proof}

Theorem \ref{thm:analytic} describes sufficient conditions for analyticity of the volume function, thereby ensuring that $\volfapprox(t;\toi)$ in \eqref{eq:taylor}
converges to $\volf(t;\toi)$ for all time. In particular, it is sufficient to ensure that the analytic continuation of the chosen regressor $\hat{f}$ is holomorphic on a polydisk containing $\uspace^n$. This property holds for many classes of algebraic functions, e.g., polynomials, trigonometric functions, etc.

\subsection{Bounding the Volume Function}
The previous section presents a means of computing a truncated Taylor expansion of the model's true volume function about $t=\toi$. However, the truncated approximation does not produce the formal bound $\bar{P}_\toi$. To compute $\bar{P}_\toi$, we employ Taylor's theorem, i.e., a bound on the error between the truncation and the true function.

Recall, Taylor's theorem states that the remainder $r^k(t) = \volf(t) - \volfapprox(t)$ of a degree-$k$ Taylor expansion is bounded
by \eqref{eq:remainder_bound}.
%with
%\begin{equation}
%    |r^k(t)| \leq M_k \frac{(t-\toi)^{k+1}}{(k+1)!}
%\end{equation}
%where $\big|\frac{d^{k+1}}{dt^{k+1}}\volf(\xi) \big| \leq M_k$ for all $\xi \in [t, \toi]$.
Let $\mathcal{R}[t_1, t_2]$ be a set containing the union of all $\cmass(t)$ in the range $t_1 \leq t \leq t_2$, i.e., 
\begin{equation} \label{eq:tube_region}
    \mathcal R[t_1, t_2] \supseteq \{\pu \text{ s.t. } \pu \in \cmass(t) \text{ for some } t \in [t_1, t_2]\}.
\end{equation}
Intuitively, $\mathcal R$ describes an over-approximation of the ``tube'' carved out by $\cmass$ as it moves through the field.
A straight-forward bound $M_k$ can be computed using the ``worst-case'' point-wise volumetric growth over $\mathcal R[t, \toi]$. Specifically, the largest possible value of, e.g. $\int_{\cmass(t)} \Gamma_{k+1}(\pu) d\pu$ for any possible $\cmass(t) \subseteq \mathcal R[t, \toi]$ is bounded by the extrema of $\Gamma_{k+1}(\pu)$ multiplied by the volume of $R[t, \toi]$. The following theorem formalizes an upper bound for $t \leq \toi$.
%By interpreting $\Gamma_k(\pu)$ as the differential point-wise time derivatives, a straight forward bound $M_k$ can be constructed using the largest \textit{point-wise} volumetric growth over $\mathcal R$.

\begin{theorem} \label{thm:bound_1}
    Suppose $\hat{f}$ satisfies Theorem \ref{thm:analytic}. 
    Let
    \begin{equation}
        \delta_k \geq \sup_{\pu \in \mathcal R[t, \tau]} \Big((-1)^{k+1}\Gamma_{k+1}(\pu)\Big).
    \end{equation}
    Then, for $t \in [0, \toi]$,
    %\begin{equation} \label{eq:bound_1}
    %    \volf(t; \toi) \leq \volfapprox(t; \toi) + \bar{r}(t) 
    %\end{equation}
    %where
    \begin{equation} \label{eq:bound_1}
        r^k(t) \leq \delta_k\vol{\mathcal R[t, \toi]} \frac{|t-\toi|^{k+1}}{(k+1)!}.
    \end{equation}
\end{theorem}
\begin{proof}
    We aim to upper-bound the Taylor-series expansion on the interval $[t, \toi]$. 
    Consider the mapping $t' = g(t)$ defined as
    \begin{equation}
        g : t \mapsto \toi - t
    \end{equation}
    which re-centers the expansion about $0$ and reverses time, reflecting the expansion function about its center. We denote the translated and reflected expansion with $\volf'(g(t);\toi)$ such that $\volf'(g(t); \toi) = \volf(t; \toi)$.
    It follows from Taylor's theorem that for $t' > 0$,
    \begin{align}
        r'^k(t') \leq M^u_k \frac{(t')^{k+1}}{(k+1)!} 
    \end{align}
    with $M^u_k \geq \frac{d^{k+1}}{(dt')^{k+1}}\volf'(\xi)$ for all $\xi \geq 0$.
    The $t'$- and $t$-derivatives are related via
    \begin{equation}
        \frac{d^{k+1}}{(dt')^{k+1}}\volf'(t') = (-1)^{k+1} \frac{d^{k+1}}{(dt)^{k+1}}\volf(t).
    \end{equation}
    Using \eqref{eq:integ},
    \begin{align}
        \sup_{\xi' \in [0, t']} \frac{d^{k+1}}{(dt')^{k+1}}\volf'(\xi') =& \sup_{\xi \in [t, \toi]} (-1)^{k+1}\frac{d^{k+1}}{(dt)^{k+1}}\volf(\xi)  \notag \\
        =& \sup_{\xi \in [t, \toi]} \int_{\cmass(\xi)} (-1)^{k+1}\Gamma_k d\pu. \label{eq:supsup}
        %\leq& \sup_{\xi \in [t, \toi]}\vol{\cmass(\xi)} \sup_{\pu \in \mathcal R[t, \tau]} \Big((-1)^{k+1}\Gamma_{k+1}(\pu)\Big).
    \end{align}
    Let $\cmass^\star$ be the maximizer of $\sup_{\xi \in [t, \toi]}\vol{\cmass(\xi)}$. Then,
    \begin{align}
        \eqref{eq:supsup} \leq \vol{\cmass^\star} \sup_{\pu \in \cmass^\star} \Big((-1)^{k+1}\Gamma_{k+1}(\pu)\Big).
    \end{align}
    It follows that $\vol{\cmass^\star} \leq \vol{\mathcal R[t, \toi]}$ and $\sup_{\pu \in \cmass^\star} (\cdot) \leq \sup_{\pu \in R[t, \toi]}(\cdot)$.
\end{proof}
\noindent
Given a bound on $\Gamma_{k+1}$ over $\mathcal R[t, \tau]$, Theorem \ref{thm:bound_1} can be used to compute $\bar{P}_\toi$. This informs the choice of $\hat{f}$, since generating formal bounds over arbitrary continuous non-linear, non-monotonic functions is known to be a challenging problem \cite{murty1985some}. Choosing $\hat{f}$ to belong to a class of functions that generate $\Gamma_k$ for which there exists a tractable means of computing tight bounds can greatly reduce conservativeness of $\bar{P}_\toi$. However, the bounds cannot be too loose, or they may not adhere to the sub-factorial growth required for the Taylor series to converge (see the proof of Theorem \ref{thm:analytic}). 

\begin{corollary} \label{cor:bound_converge}
    Suppose $\hat{f}$ satisfies Theorem \ref{thm:analytic}.
    Let 
    \begin{equation}
        e_k = \delta_k - \sup_{\pu \in \mathcal R[t, \tau]} \Big((-1)^{k+1}\Gamma_{k+1}(\pu)\Big).
    \end{equation}
    Then, if 
    \begin{equation} \label{eq:bound_converge}
        \lim_{k\rightarrow \infty} \frac{e_k}{(k + 1)!} = 0,
    \end{equation}
    the remainder bound \eqref{eq:bound_1}
    converges to zero from above.
\end{corollary}
\begin{proof}
    The statement is a straight-forward consequence of Theorem \ref{thm:analytic}, with the caveat that the bounding procedure does not increase in conservativeness quicker than the factorial.
\end{proof}
\noindent
Corollary \ref{cor:bound_converge} explicitly outlines the relationship between the bounding procedure and upper-bound convergence. Basically, it is important to verify that for higher-order expansions, the bound must not become too conservative.

The bound in Theorem \ref{thm:bound_1} can be further refined for situations where $\delta_k > 0$. 
Theorem \ref{thm:bound_1} uses $\vol{\mathcal R[t, \toi]}$ as a safe (but potentially conservative) bound over the largest possible volume $\cmass$ can attain throughout the interval $[t, \toi]$. This estimate can be improved by recognizing that $\volf(t; \toi)$ \textit{itself} is estimating that same quantity. Using Theorem \ref{thm:bound_1} to obtain an upper bound $V_{max} \geq \volf(t; \toi)$ over the interval $[t, \toi]$ can be used in place of $\vol{\mathcal R[t, \toi]}$ to achieve a tighter bound on the remainder. This process can be repeated an arbitrary number of times, generating a convergent sequence of bounds. 
Formally, let $\{V_{max,i}\}_{i\geq 0}$ be an infinite sequence such that $V_{max, 0} = \vol{\mathcal R[t, \toi]}$ and 
\begin{equation} \label{eq:vmax_recur}
    V_{max, i+1} = \sup_{\xi \in [t, \toi]} \Big( \volfapprox(\xi; \toi) + \delta_k V_{max, i} \frac{|\xi-\toi|^{k+1}}{(k+1)!} \Big).
\end{equation}
Under certain conditions on $\volfapprox(t; \toi)$, we show that $V_{max, i}$ is a convergent geometric series, with a known closed-form limit-value.

\begin{theorem} \label{thm:bound_2}
    Let 
    $\kappa \geq \volfapprox(\xi)$ for all $\xi \in [t, \toi]$. Then, for $0\leq \delta_k \leq \frac{(k+1)!}{|t - \toi|^{k+1}}$, 
    \begin{equation} \label{eq:bound_2}
        \sup_{\xi \in [t, \toi]} \volf(t; \toi) \leq \kappa \Big(1-\delta_k \frac{|t-\toi|^{k+1}}{(k+1)!} \Big)^{-1}
    \end{equation}
    Moreover,
    \begin{equation} \label{eq:volf_bound_2}
        \volf(t; \toi) \leq \volfapprox(t, \toi) + \kappa \Big(\frac{(k+1)!}{|t-\toi|^{k+1}} - \delta_k \Big)^{-1}.
    \end{equation}
\end{theorem}
\begin{proof}
    Following \eqref{eq:bound_1} in Lemma \ref{thm:bound_1}, we can define the $i$-th remainder bound function as
    \begin{equation}
        \bar{r}^k_i(t) = \delta_k V_{max, i}\frac{|t-\toi|^{k+1}}{(k+1)!}.
    \end{equation}
    For $\delta_k \geq 0$, each $\bar{r}^k_i(t)$ is monotonically decreasing on $[t, \toi]$, and thus, is maximized at $t$. Therefore \eqref{eq:vmax_recur} simplifies to
    \begin{align}
        V_{max, i+1} &\leq sup_{\xi \in [t, \toi]} \Big(\kappa + \delta_k V_{max, i}\frac{|\xi-\toi|^{k+1}}{(k+1)!} \Big) \notag \\
        &= \kappa + \delta_k V_{max, i}\frac{|t-\toi|^{k+1}}{(k+1)!}.
    \end{align}
    Let $\alpha = \delta_k \frac{|t-\toi|^{k+1}}{(k+1)!}$. Then $V_{max, i+1} = \kappa + \alpha V_{max, i}$, which, in non-recursive form, generates the series
    \begin{equation}
        V_{max, i} = \sum_{j=0} \kappa \alpha^j
    \end{equation}
    which, for $\alpha \in [0, 1)$ converges to the RHS of \eqref{eq:bound_2} as $i\rightarrow \infty$. The bound \eqref{eq:volf_bound_2} follows from replacing $\vol{\mathcal R[t, \toi]}$ in \eqref{eq:bound_1} with the converged bound in \eqref{eq:bound_2}.
\end{proof}

While the bound in Theorem \ref{thm:bound_2} appears to always dominate Theorem \ref{thm:bound_1}, it requires a small overhead in the computation of 
$\kappa$, i.e., a bound over a univariate polynomial\footnote{There are many ways to compute the bound over a univariate polynomial. For our experiments, we used a Bernstein-basis conversion (see \eqref{eq:bern_bounds}).}. Additionally, if the estimate expansion $\volfapprox$ decreases substantially over $[t, \toi]$, $\kappa$ may be large enough such that Theorem \ref{thm:bound_2} is more conservative than Theorem \ref{thm:bound_1}. However, in practice, the remainder bound often grows much faster than $\volfapprox$, and $\volfapprox$ generally does not vary much for small $\toi - t$, making Theorem \ref{thm:bound_2} less conservative.

%The key insight revolves around the circularity of computing an upper-bound on the volume of $\cmass(t)$ that \textit{itself} depends on the maximum encountered volume of $\cmass(t)$ on the interval $[t, \toi]$. In the proof of Lemma \ref{lem:bound_1}, $\vol{\mathcal R[t, \toi]}$ is used as a safe upper bound for this volume. However, since the bound on $r^k(t)$ can be used to compute a bound on the true volume $\volf{t; \toi}$ over $[t, \toi]$, the 

\section{Choosing a Regressor} \label{sec:model}

We now focus on the model learning piece of Problem~1.\ref{prob:learn}.
In order to choose a parameterized regressor $\hat{f}$, we first recall the properties that $\hat{f}$ must possess for asymptotic universal convergence and uncertainty propagation. 
Let each $\hat{f}_i$ belong to a parameterized class of functions $\mathcal F_\theta \subset \conts{\uspace^n}$. Then, $\hat{f}$ must satisfy the following properties:
\begin{enumerate}[label=(P\arabic*)]
    \item each $\hat{f}_i$ is a convergent universal estimator (Def. \ref{def:convergent}), \label{cond:cue}
    \item boundary conditions \eqref{eq:boundary_conds} can be exactly enforced, \label{cond:bcs}
    \item each $\hat{f}_i$ satisfies Theorem \ref{thm:analytic}, \label{cond:analytic}
    \item $\Gamma_k$ has an analytical antiderivative for all $k > 0$, and \label{cond:antideriv}
    \item $\Gamma_k$ can be upper and lower bounded such that \eqref{eq:bound_converge} holds. \label{cond:bound}
\end{enumerate}
Each property poses a unique restriction over $\mathcal F_\theta$.
In this work, we study the Bernstein polynomial regressor and show that it satisfies conditions \ref{cond:cue}-\ref{cond:bound}.

A (multivariate) Bernstein polynomial $b(\pu) \in \conts{\uspace^n}$ of degree $\mathbf{d}=(d_1, \ldots, d_n)$ is defined as
\begin{equation}
    b(\pu) = \sum_{\mathbf{0} \leq \mathbf{j} \leq \mathbf{d}} \theta_{\mathbf{j}} \prod_{i=1}^n \beta_i(u_i; \mathbf{j}, \mathbf{d}),
\end{equation}
where $\mathbf{j}=(j_1, \ldots, j_n)$ is a multi-index, and
\begin{equation}
    \beta_i(u_i; \mathbf{j}, \mathbf{d}) = \binom{j_i}{d_i} u_i^{d_i} (1-u_i)^{d_i - j_i}
\end{equation}
are the Bernstein basis functions. Bernstein polynomials describe a polynomial basis, i.e., every (maximal) degree $\mathbf{d}$ polynomial can be represented as a Bernstein polynomial and visa-versa. Each component $\hat{f}_i$ is learned as an independent multivariate Bernstein polynomial.
Below, we show how the Bernstein polynomial regressor satisfies each of the aforementioned conditions.

\paragraph{Condition \ref{cond:cue} (Convergent Universal Estimator)}
The Weierstrass approximation theorem \cite{stone1948generalized} states that (Bernstein) polynomials are universal approximators. In fact, Bernstein polynomials are known to have strong convergence with respect to to the $\mathcal C^k$-norm \cite{veretennikov2016partial}, as required by Def. \ref{def:convergent}.
Using the Bernstein polynomial basis as polynomial features in linear regression ensures that the loss landscape is convex. Hence, as the number of features increase, i.e., $\mathbf{d}$ increases, the global minimizer (in the least-squares sense) is both computable and convergent to the true system.

\paragraph{Condition \ref{cond:bcs} (Boundary Conditions)}
Recall, for $\hat{f}$ to be a valid $\pu$-space vector field, the perpendicular components of $\hat{f}$ along the boundary of $[0,1]^n$ must be zero. Bernstein polynomials offer a simple means of enforcing such boundary conditions. Specifically, for a multi-index $\mathbf{j} = (j_1, \ldots, j_n)$, if $\theta_{\mathbf{j}} = 0$ when $j_i=0$ or $j_i=d_i$, then $\hat{f}_i(\pu) = 0$ when $u_i=0$ or $u_i=1$. 

\paragraph{Condition \ref{cond:analytic} (Analytic Volume Function)} Polynomials are known to be \textit{entire}, i.e., the analytic extension to the complex plane is holomorphic for all $\mathbb{C}^n$. Therefore, the conditions in Theorem \ref{thm:analytic} are easily satisfied, and the Bernstein polynomial regressor generates an analytic volume function.

\paragraph{Condition \ref{cond:antideriv} (Closed-form Antiderivative)} Polynomials always possess an easy-to-compute anti-derivative. Expressed in the Bernstein basis, the antiderivative can be computed as a linear combination of the coefficients, or the incomplete beta function \cite{lorentz2012bernstein}.

\paragraph{Condition \ref{cond:bound} (Bound Convergence)} Bernstein polynomials offer an efficient and tight bounding procedure. Specifically, any Bernstein polynomial $b(\pu)$ is bounded by its coefficients, i.e.
\begin{equation} \label{eq:bern_bounds}
    \min_{\mathbf{0} \leq \mathbf{j} \leq \mathbf{d}} \theta_{\mathbf{j}} \leq \min_{\pu \in \uspace^n} b(\pu) \leq \max_{\pu \in \uspace^n} b(\pu) \leq \max_{\mathbf{0} \leq \mathbf{j} \leq \mathbf{d}} \theta_{\mathbf{j}}.
\end{equation}

Additionally, \cite{garloff1985convergent} outlines procedures for generating arbitrarily tight bounds, thus satisfying Corollary \ref{cor:bound_converge}. In practice, however, such bounding procedures are often not needed and \eqref{eq:bern_bounds} is sufficient. 

In addition to satisfying the aforementioned properties, algebraic operations (such as multiplication, differentiation, and integration used in Lemma \ref{lem:v_derivs}) are relatively numerically stable, and can be performed in the Bernstein basis \cite{murty1985some, lorentz2012bernstein}. Note that the aforementioned criteria are not solely restricted to Bernstein polynomials, and rather outline the dimensions along which the efficacy of the proposed method is effective for a given regressor. Investigating other regressors and their adherence to the criteria is a valuable direction for future work.

%\begin{figure}[!ht]
%    \centering
%    \includegraphics[width=0.5\textwidth, trim={380 250 80 50}, clip]{figures/transport_poly_figure.pdf}
%    \caption{Example ``tamed'' expansion: the expansion derivatives computed from $l=1$ flowpipe region are formally adjusted by $\Delta$, then compared to the derivatives predicted by the $l=0$ expansion.}
%    \label{fig:ex-tame}
%\end{figure}

The theoretical properties of Bernstein polynomials enable a strong asymptotic convergence result for the proposed joint modeling and propagation procedure, which is a straight-forward consequence of the satisfaction of conditions \ref{cond:cue}-\ref{cond:bound}.
\begin{theorem} \label{thm:converge}
    Suppose each $\hat{f}_i$ be a degree-$\mathbf{d}$ Bernstein polynomial estimator subject to the boundary conditions \eqref{eq:boundary_conds}. Let $\bar{P}^k_\toi(\hat{f})$ denote the bound computed using Theorem \ref{thm:bound_1} for a given model $\hat{f}$ and a degree-$k$ expansion. Then, 
    \begin{equation}
        \lim_{\mathbf{d}, |\mathcal{D}|, k \rightarrow \infty} |\bar{P}^k_\toi(\hat{f}) - P(\phi(\px_0, \toi) \in R)|=0,
    \end{equation}
    i.e., $\bar{P}^k_\toi(\hat{f})$ converges to the true system's probability.
\end{theorem}
%\begin{proof}
%    
%\end{proof}
\noindent
Theorem \ref{thm:converge} illuminates the convergence properties of the proposed uncertainty propagation method when using a Bernstein polynomial regressor.

\section{Backward Set Propagation}
Sections \ref{sec:prop} and \ref{sec:model} describe a means of generating asymptotically convergent bounds on the model's predictions. However, in practice, truncated Taylor approximations often only provide accurate estimates within a local neighborhood of the expansion point, even if the underlying function has infinite radius of convergence. To address this problem, we can leverage reachable set propagation \cite{chen2013flow, bogomolov2019juliareach, althoff2015introduction} to maintain a crude, mathematically simple over-approximation of the true shape of $\cmass(t)$. The reachable set propagation generates a flowpipe, i.e., a set of transition regions that are guaranteed to contain $\cmass(t)$ over a finite number of time intervals. Referencing Fig. \ref{fig:ex-deformation}, we can compute a flowpipe that starts at $t=\toi$, and over-approximates $\cmass(t)$ as it travels backwards in time to $t=0$. 

Formally, a flowpipe is a sequence of tube regions $S=\{\mathcal R_l[t_l, t_{l+1}]\}$ where $l \in \{0, \ldots, L\}$ such that each $\mathcal R_l$ satisfies \eqref{eq:tube_region}, $t_{l=0}=\toi$, and $t'_L=0$. Additionally, the flowpipe theory also provides point-wise set over approximations at specific instances $t$, denoted $\mathcal R[t] \supseteq \cmass(t)$.
We assume that each $\mathcal R_l$ is hyperrectangular. Since it is trivial to calculate the volume of each $\mathcal R_l$, the flowpipe itself can be used to calculate a very conservative $\bar{P}_\toi$. For simplicity, let $\volfbound$ denote the approximate expansion plus the bound monomial, making it an upper bound on $\volf$.

We combine the computational practicality of a box-flowpipe with the local accuracy of the Taylor volume function expansion. For each flowpipe box, we can build a Taylor expansion $\volfbound^l(t; t_l)$ centered about $t_l$ for estimating the volume on the interval $[t_l, t_{l+1}]$. Each Taylor expansion is created by integrating over the hyperrectangular over-approximation $\mathcal R_l[t_l]$. This alone, does not meaningfully reduce the conservativeness compared to just using the volume of each $\mathcal R_l$. However, we can use the computed coefficients of each $\volfbound^{l+1}(t; t_l)$ to ``tame'' the diverging derivatives of the previous expansion $\volfbound^{l}(t; t_l)$.

Consider (without loss of generality) the expansions of the first and second flowpipe regions: $\volfbound^0(t; t_0=\toi)$ constructed by integrating over $\mathcal R_0[t_0] = R$ and $\volfbound^1(t; t_1)$ constructed by integrating over $\mathcal R_1[t_1]$. 
%A visual example is shown in Fig. \ref{fig:ex-tame}. 
At time $t=t_1$, the derivatives of the previous expansion upper-bound all $k$ time derivatives, i.e.
\begin{align}
    \frac{d^k \volf}{dt^k}\Big|_{t = t_1} &\leq \frac{d^k \volfbound^0}{dt^k}\Big|_{t = t_1}. \label{eq:prev_estimate}
\end{align}
On the other-hand, we can construct an alternative (upper-bound) estimate using the integral of $\Gamma_k$ over the over-approximate rectangular region $R_1[t_1]$. To see how, consider the decomposition $\Gamma_k = \Gamma_k^+ + \Gamma_k^-$ where $\Gamma_k^+ \geq 0$, $\Gamma_k^- \leq 0$, such that $\Gamma_k^+ \geq \Gamma_k$ and $\Gamma_k^- \leq \Gamma_k$. Since $\Gamma_k$ is a Bernstein polynomial, $\Gamma_k^+$ can be created by setting all negative coefficients of $\Gamma_k$ to zero. We aim to upper bound $\int_{\cmass(t_1)} \Gamma_k(\pu) d\pu$, however, at $t_1$ (or any future time $t_l > t_0$), we do not know exactly where $\cmass(t_1)$ is, only that it must be contained in $R_1[t_1]$. Therefore, $\int_{\cmass(t_1)} \Gamma_k(\pu) d\pu$ is upper-bounded for \textit{any} possible $\cmass(t_1) \subseteq R_1[t_1]$ as follows
\begin{align}
    \frac{d^k \volf}{dt^k}\Big|_{t = t_1} = \int_{\cmass(t_1)} \Gamma_k(\pu) d\pu =& \int_{\cmass(t_1)} \Gamma^+_k(\pu) + \Gamma^-_k(\pu) d\pu \notag \\
    \leq& \int_{\cmass(t_1)} \Gamma^+_k(\pu) d\pu \notag \\
    \leq& \int_{R_1[t_1]} \Gamma^+_k(\pu) d\pu. \label{eq:curr_estimate}
\end{align}
%On the other hand, we can construct an alternative (upper-bound) estimate by integrating the non-negative part of $\Gamma_k$ (denoted $\Gamma^+_k$) over $R_1[t_1]$, effectively upper-bounding \eqref{eq:integ} for any $\cmass(t=t_1)\subset R_1[t_1]$.
%On the other hand, constructing $\volfapprox^1(t; t_1)$ by integrating each $\Gamma_k$ over $R_1[t_1]$ (using \eqref{eq:integ}) yields an alternative set of ``estimates'' of the current time derivatives. Unfortunately, the derivatives are not directly comparable because they may not correspond to the same region of integration. 
The bound in \eqref{eq:prev_estimate} will be tighter if the previous local Taylor expansion is accurate, whereas \eqref{eq:curr_estimate} will be tighter if the reachable set approximation is more accurate (i.e. the previous Taylor series is diverging). Both bounds are sound over-approximations of the true derivatives, therefore, the smaller one is selected for each coefficient of the expansion of $\volfapprox^0(t; t_1)$. This procedure is iterated over each set in the flowpipe, as described in the following Algorithm.

\begin{figure*}[!ht]
    \centering

    %========================
    % Row a: (a.1), (a.2)
    %========================
    \setcounter{subfigure}{0}%
    \renewcommand\thesubfigure{a.\arabic{subfigure}}%

    \makebox[0pt][r]{\rotatebox{90}{\qquad \qquad \qquad \qquad Probability}\hspace{-2.5em}}
    \begin{subfigure}{0.48\textwidth}
        \centering
        \includegraphics[width=0.8\textwidth, trim={0 0 0 0}, clip]%
            {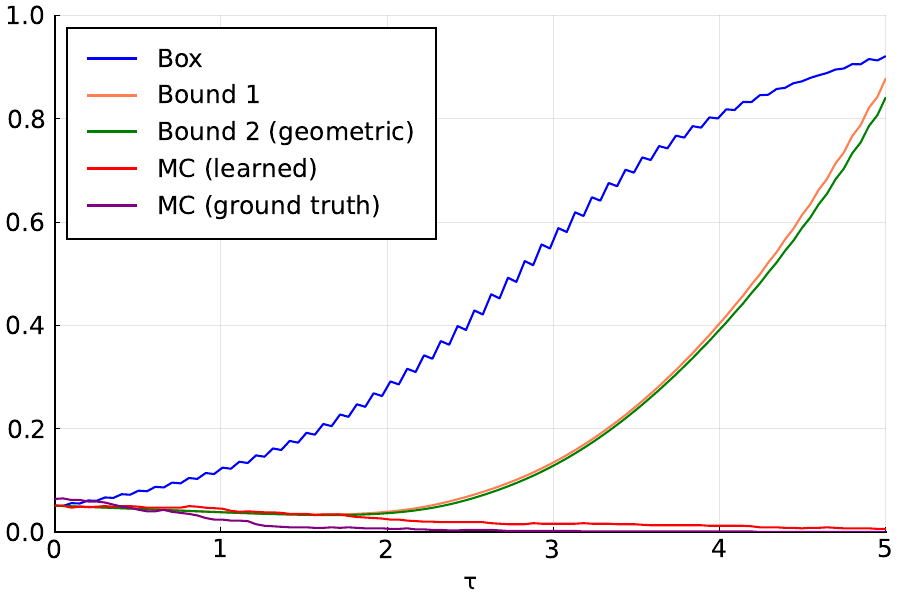}
        \caption{Probability vs.\ $\toi$}
        \label{fig:exp:vanderpol_prob} % -> (a.1)
    \end{subfigure}
    \hfill
    \begin{subfigure}{0.48\textwidth}
        \centering
        \includegraphics[width=0.9\textwidth, trim={0 0 0 0}, clip]%
            {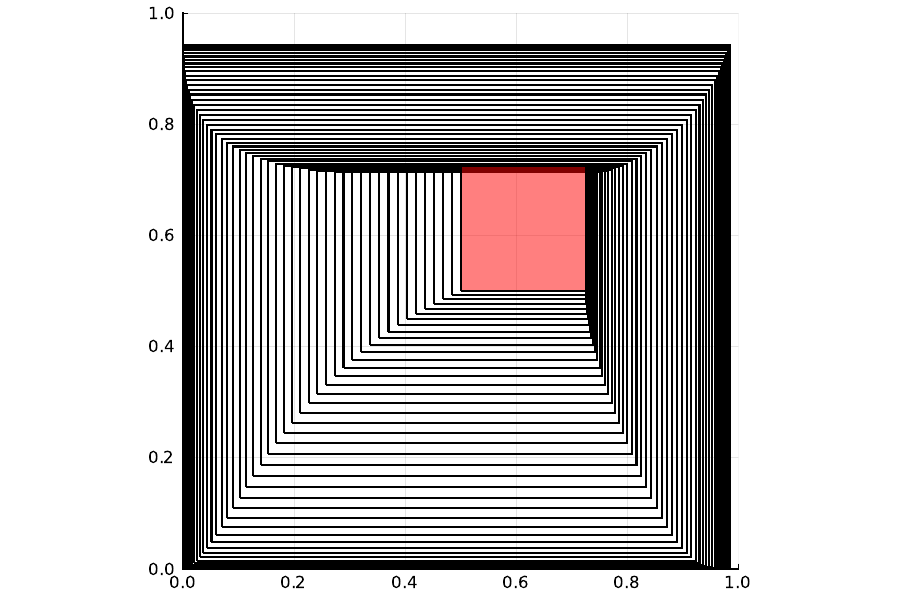}
        \caption{Flowpipe ($\uspace^n$) ($R$ shown in red)}
        \label{fig:exp:vanderpol_fp} % -> (a.2)
    \end{subfigure}

    \vspace{1em}

    %========================
    % Row b: (b.1), (b.2)
    %========================
    \setcounter{subfigure}{0}%
    \renewcommand\thesubfigure{b.\arabic{subfigure}}%

    \makebox[0pt][r]{\rotatebox{90}{\qquad \qquad \qquad \quad Log-probability}\hspace{-2.5em}}
    \begin{subfigure}{0.48\textwidth}
        \centering
        \includegraphics[width=0.8\textwidth, trim={0 0 0 0}, clip]%
            {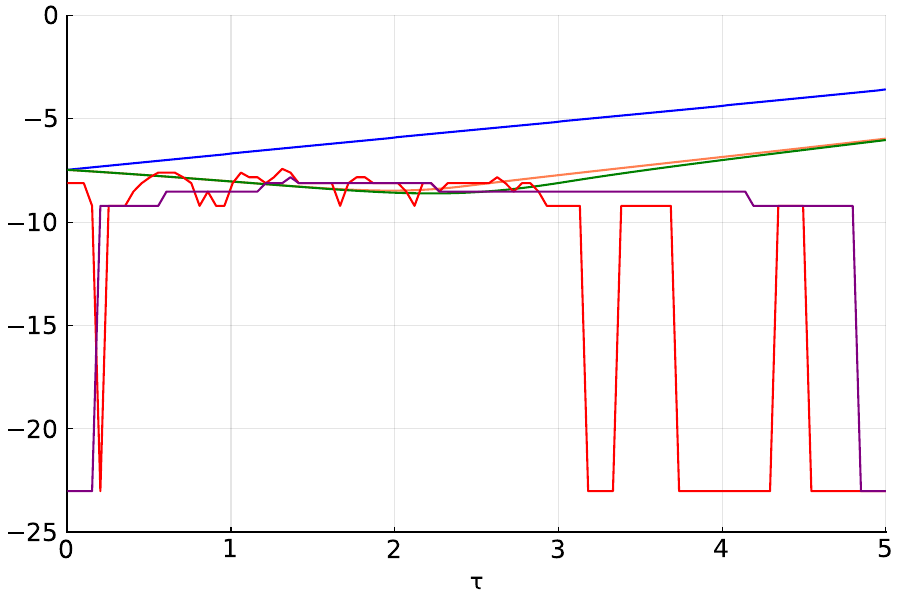}
        \caption{Log-Probability vs.\ $\toi$}
        \label{fig:exp:vanderpol_rare_prob} % -> (b.1)
    \end{subfigure}
    \hfill
    \begin{subfigure}{0.48\textwidth}
        \centering
        \includegraphics[width=0.9\textwidth, trim={0 0 0 0}, clip]%
            {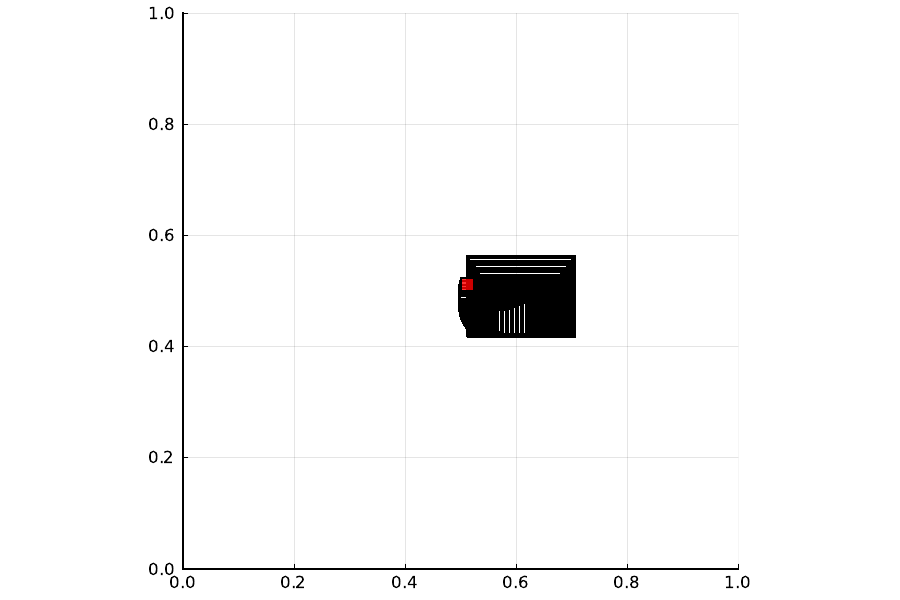}
        \caption{Flowpipe ($\uspace^n$)}
        \label{fig:exp:vanderpol_rare_fp} % -> (b.2)
    \end{subfigure}

    \caption{2D Van der Pol: probability and flowpipe evolution (top, a.\*), and rare-event counterparts (bottom, b.\*).
    }
    \label{fig:exp:vanderpol_all}
\end{figure*}

\begin{algorithm}
\caption{Flowpipe-Tamed Volume Function Expansion}
\label{alg:tamed_ts}
\begin{algorithmic}[1]
    \STATE \textbf{Input:} Flowpipe $S$, Expansion degree $m$
    %\STATE 
    %\STATE Create $\volfapprox^{l}(t; t_{l})$ with coefficients $v_{k, init}$, plus bound monomial
    \FOR{$l = 1$ \TO $L-1$}
        %\STATE Compute derivatives $v_{k, curr} = \frac{d}{dk}\volfapprox^{l}(t; t_{l})\big|_{t=t_l}$
        %\STATE Compute $\volfapprox^l(t; t_l=)$ using integration region $\mathcal R[t_l]$ and bounding region $\mathcal R[t_l, t'_l]$
        \IF{$l = 0$}
            \STATE $v_{k, curr} \gets \int_{\mathcal R} \Gamma_k(\pu)d\pu$ $\forall k \in \{1, \ldots, m\}$
        \ELSE
            \STATE $v_{k, prev} \gets \frac{d}{dk}\volfapprox^{l-1}(t; t_{l-1})\big|_{t=t_l}$ $\forall k \in \{1, \ldots, m\}$
            \STATE $v_{k, new} \gets \int_{\mathcal R[t_l]} \Gamma^+_k(\pu)d\pu$ $\forall k \in \{1, \ldots, m\}$
            \STATE $v_{k, curr} \gets \min(v_{k, new}, v_{k, prev})$ $\forall k \in \{1, \ldots, m\}$
        
        \ENDIF
        \STATE Create bound monomial over $\mathcal R[t_l, t_{l+1}]$
        \STATE Create $\volfbound^{l}(t; t_{l})$ with coefficients $v_{k, curr}$, plus bound monomial
    \ENDFOR
    \STATE \textbf{Output:} $\{\volfbound^{l}(t; t_{l}) \}_{l \in \{1, \ldots, L\}}$
\end{algorithmic}
\end{algorithm}

Algorithm \ref{alg:tamed_ts} outlines the procedure for creating a piece-wise collection of ``tamed'' volume function expansions. The algorithm takes as input a flowpipe $S$ and an expansion degree. 
For the first flowpipe box, i.e. $R$, the expansion is created using the exact integrals (line 4). For each subsequent flowpipe box, the derivative estimates from the previous expansion are calculated using \eqref{eq:prev_estimate} (line 6), and the ``new'' derivative estimates are calculated by integrating over the current flowpipe box using \eqref{eq:curr_estimate} (line 7). The less conservative derivatives are chosen (line 8). 
The chosen derivatives are used in the coefficients of \eqref{eq:taylor}, and a bound monomial (computed using Theorem \ref{thm:bound_1} or Theorem \ref{thm:bound_2}) is computed over the transition region, forming the next expansion (line 10-11).
%and the next expansion is constructed (line 10-11). 
%The algorithm outputs a piece-wise time polynomial.
%%For each flowpipe box, all $m$ derivatives are initially calculated (line 3). If $l \geq 1$, there exists a previous predictor to compare to. 
%The derivatives (evaluated at the new time $t_l$) of the previous predictor are calculated (line 5), and the current predictor's derivatives are worst-case adjusted (line 6). The adjusted derivatives are compared to the previous predictors derivatives, and the smaller one is selected (line 7). 

\section{Experimental Results} \label{sec:experiments}
This section evaluates the practical utility of the proposed modeling and propagation method. 
%We examine two empirical case studies: i) propagation results for learning a 2D and 4D system, ii) ablation studies with regard to the expansion degree.

\paragraph{Experimental Setup}
We performed experiments on two systems, a 2D Van der Pol, and a 4D cartpole. The Van der Pol system is described by 
\begin{align}
    \dot{x}_1 &= x_2 \notag \\
    \dot{x}_2 &= x_2 (1 - x_1^2) - x_1
\end{align}
with initial distribution $\mathcal{N}(\mathbf{0}, 0.5^2 I_{2\times2})$. The cartpole system is described by $\px = [y, v_y, \theta, \omega]$
\begin{align}
    \dot{y} &= v_y \notag \\
    \dot{v}_y &= \frac{m_p \sin(\theta) (l\omega^2 + g \cos(\theta))}{m_c + m_p \sin^2(\theta)} \notag \\
    \dot{\theta} &= \omega \notag \\
    \dot{\omega} &= \frac{m_p l \omega^2 \cos(\theta)\sin(\theta) - (m_c + m_p) g \sin(\theta)} {l(m_c + m_p \sin^2(\theta))}
\end{align}
where $g=9.81, l=1.0, m_p=0.1$, and $m_c=1.0$ and the inital distribution is $\mathcal{N}([0.0, 0.0, 0.1, 0.5], \text{diag}(0.5^2, 0.1^2, 0.2^2, 0.4^2))$.
Three bounding methods are compared against Monte-Carlo forward Euler-integration simulations with 10k samples for both the learned model and the true underlying system (ground truth). The first bounding method (``Box'') builds successive taylor for each flowpipe box. Then, we compare both bounding methods, i.e. Theorem \ref{thm:bound_1} (``Bound 1'') and Theorem \ref{thm:bound_2} (``Bound 2''), using the tamed expansion procedure in Algorithm \ref{alg:tamed_ts}. For the Van der Pol system, the model was regressed as a degree $7$ Bernstein polynomial, and an expansion degree $m=5$ was chosen. For the cartpole system, the model is degree $6$ and $m=4$.

\subsection{Propagation Results}

Figure~\ref{fig:exp:vanderpol_prob} shows the results for a regularly sized $R$. The box method provides the most conservative $\bar{P}_\toi$, since it roughly captures the conservatism in the flowpipe itself, which can be seen in Figure~\ref{fig:exp:vanderpol_fp}. Up until $\toi=2.0$, the tamed methods provide a very tight $\bar{P}$. Beyond $\toi=2.0$ the flowpipe becomes more conservative, the ``taming'' effect diminishes. This means that the flowpipe communicates less informative information when taming the expansion derivatives, and thus the expansion becomes more conservative, eventually catching up to the box method.

Figure \ref{fig:exp:vanderpol_rare_prob} shows the results for a very small $R$, i.e. a \textit{rare-event}. Note that the probabilities in Fig. \ref{fig:exp:vanderpol_rare_prob} are shown using a log-scale. Since the event is so rare, very few forward simulations end up in $R$, making the probability estimate highly inaccurate, even for 10k samples. Conversely, the tamed expansion methods are able to provide a very tight and informative bound until $\toi=3.0$. This experiment highlights the power of the proposed method, showing that, unlike sampling-based simulation, rarity of the target event does not impact the predictive performance.

The results for the 4D cartpole system are shown in Fig. \ref{fig:exp:cartpole}. 
The resulting behavior is very similar to that of the 2D Van der Pol system. 
Specifically, there is a period where the tamed method provides a very tight bound, then, as the flowpipe becomes highly conservative, the tamed method becomes less accurate. 
The geometric bound is able to hold a tighter estimate for longer before diverging, illuminating the benefit of using Theorem \ref{thm:bound_2} over Theorem \ref{thm:bound_1}.
%Unlike the Van der Pol experiment, the expansions computed via the box method show a local decreasing trend at each $t_l$, even though the flowpipe still grows (explaining the small jumps). 
%Since the expansion is locally decreasing, that information information is communicated to the tamed methods. 
%This allows the tamed methods to maintain an improved bound over the box method, instead of converging to the box method's bound.
The results for the rare event (Fig. \ref{fig:exp:cartpole_rare}) are consistent with the Van der Pol experiment, yielding very accurate bounds for a significant period of time compared to Monte Carlo simulations.

\begin{figure*}
    \centering
    \makebox[0pt][r]{\rotatebox{90}{\qquad \qquad \qquad \qquad Probability}\hspace{-1.5em}}
    % First (left) subfigure
    \begin{subfigure}{0.48\textwidth}
        \centering
        \includegraphics[width=0.9\textwidth, trim={0 0 0 0}, clip]{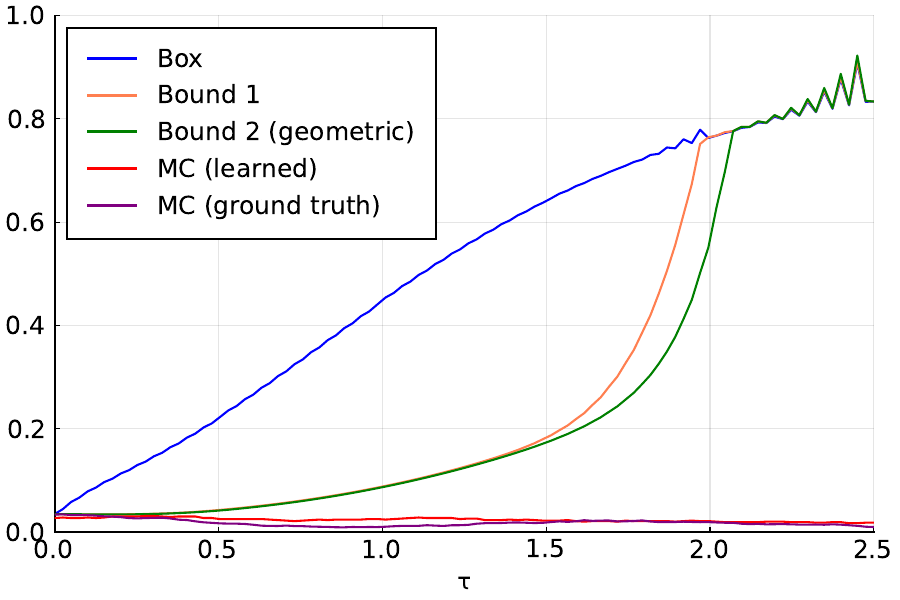}
        \caption{Probability vs. $\toi$}
        \label{fig:exp:cartpole_prob}
    \end{subfigure}
    \hfill
    \makebox[0pt][r]{\rotatebox{90}{\qquad \qquad \qquad \qquad Log-Probability}\hspace{-1.5em}}
    % Second (right) subfigure
    \begin{subfigure}{0.48\textwidth}
        \centering
        \includegraphics[width=0.9\textwidth, trim={0 0 0 0}, clip]{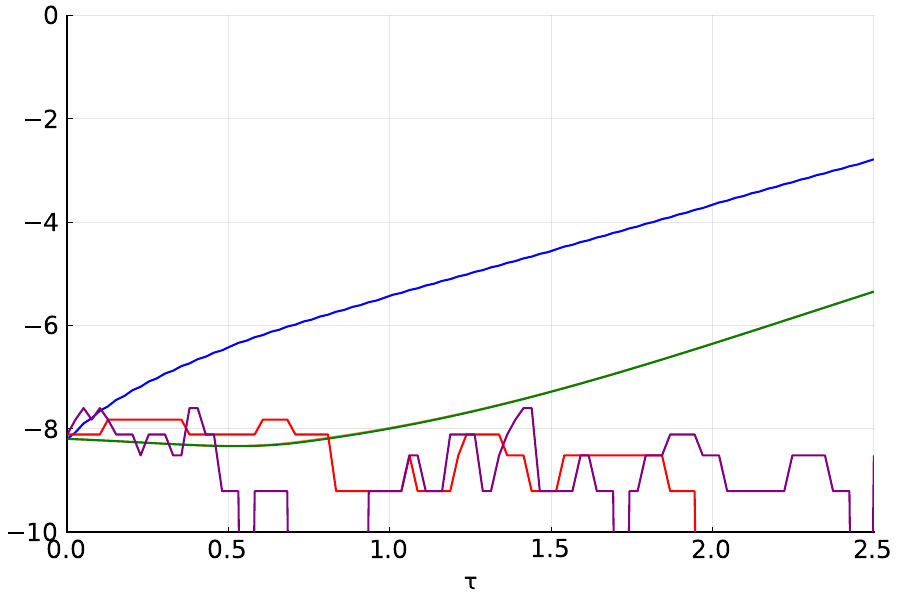}
        \caption{Log-probability vs. $\toi$ (rare event)}
        \label{fig:exp:cartpole_rare}
    \end{subfigure}

    \caption{4D Cartpole}
    \label{fig:exp:cartpole}
\end{figure*}

\begin{remark} \label{rem:shape_fidelity}
    State of the art reachable set tools \cite{chen2013flow, bogomolov2019juliareach, althoff2015introduction} use higher fidelity approximations, e.g., zonotopes, of the shape of the region, or set-boundary propagation methods \cite{xue2016reach}. This significantly reduces the ``wrapping'' effect and makes the flowpipe much less conservative, even if hyperrectangular over-approximations are used in post-process. The reachable set implementation used for our experiments leverages Bernstein polynomial bounds instead of traditional interval-arithmetic bounds, but does not use high fidelity shape approximations. Embedding Bernstein bounds state of the art reachable set tools can significantly increase horizon of accurate predictions.
\end{remark}

\subsection{Expansion Degree Ablation Study}
\begin{figure}
    \centering
    \makebox[0pt][r]{\rotatebox{90}{\qquad \qquad \qquad Probability}\hspace{0.0em}}
    \includegraphics[width=0.9\linewidth]{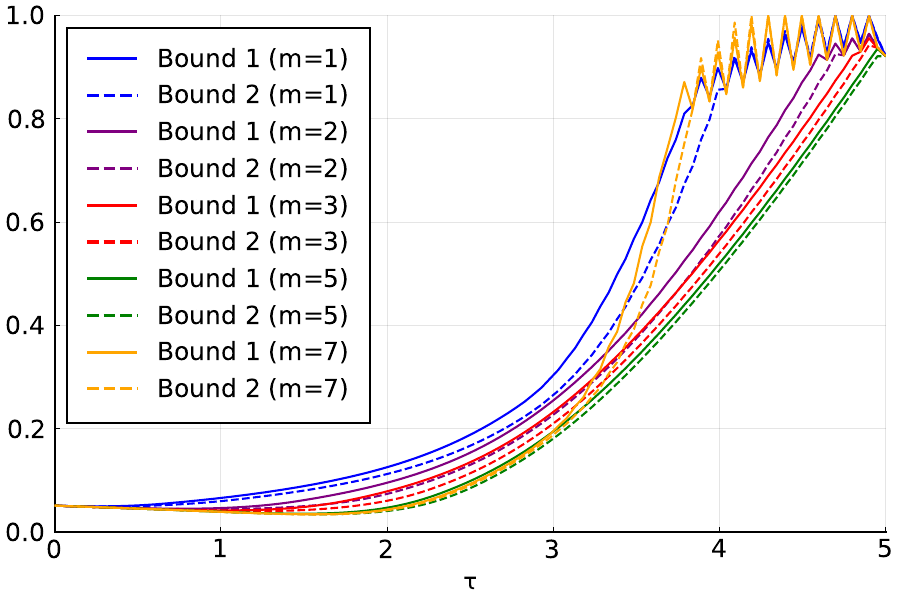}
    \caption{Expansion Degree $m$ Ablation Study}
    \label{fig:k_ablation}
\end{figure}

Figure~\ref{fig:k_ablation} shows an ablation study on the Van der Pol system varying the expansion degree $k$. As expected, increasing $m$ yields a tighter bound. However, when $m$ is large, e.g., $m=7$, the bound-monomials (e.g., \eqref{eq:bound_1}) become higher degree, and thus there is more potential to diverge as time progresses. This is consistent with Taylor series approximation, showing that higher expansion degrees improve local neighborhood estimates, but can cause sharper divergence outside of such neighborhoods. 

Additionally, since the Van der Pol system has chaotic behavior as the states become large in magnitude, flowpipe regions that occupy most of the state space may have sharply decreasing volume functions, but still must grow in order to guarantee containing all of the states. This explains the oscillatory behavior observed when the probability estimates go beyond $0.8$.

\section{Discussion}
The proposed method leverages a special domain transformation that alleviates the need to track high-fidelity estimates of the \textit{shape} of $\cmass(t)$, and rather, only precisely track the volume. Coupling this idea with a crude approximation of the shape (via reachable sets) can yield very tight probability bounds for significant periods of time. This approach generates formal certificates with respect to the learned model, enhancing the understanding and trust of the model's predictions. Additionally, the proposed method improves in performance with higher-fidelity shape approximations (see Remark \ref{rem:shape_fidelity}). Since our method ``post-processes'' a flowpipe, integration with future advancements in reachable set tools and methods is straight-forward.

The proposed method may have applications in uncertainty quantification for rare events. For approximation only (no formal bounds), using Pade approximants \cite{baker1961pade} may yield a significantly more accurate $\volfapprox$, since rational functions do not necessarily diverge like high-degree polynomials. Additionally, discretization methods for enhancing the predicted shape of $\cmass(t)$ is another promising direction of future work.

\paragraph{Limitations}
While the conditions outlined in Section~\ref{sec:model} do not assume a Bernstein polynomial regressor, there may be very few classes of regressors that satisfy all conditions. Since the model is learned from data, there always exists statistical learning \textit{error}; however, this work does not formally quantify this error with finite sample guarantees. For applications where \textit{end-to-end} safety certificates are desired, we recognize this as a valuable future direction. Furthermore, the number of coefficients in a Bernstein polynomial is exponential in the dimension. The memory required for $\Gamma_k$ is $\mathcal{O}((kd)^n)$, making the proposed method memory intensive for high-dimensional systems or large degree expansions.

\section{Conclusion}
This work tackles the joint modeling and uncertainty propagation/evaluation problem for nonlinear continuous time systems subject to initial state uncertainty. The proposed method couples continuum dynamics with Taylor series approximation theory to yield formally guaranteed probability bounds with respect to the learned model, as well as asymptotic convergence to the true system's probability. The theoretical results provide a rigorous checklist for selecting a model such that the mathematical guarantees hold. The chosen Bernstein polynomial regressor, coupled with the proposed propagation approach, shows promising results for generating model-based certificates, particularly for rare-events.

\bibliographystyle{acm}
\bibliography{bibliography}


\end{document}